\documentclass[12pt]{article}
\usepackage{amsmath}
\usepackage{amsthm}
\usepackage{amssymb}
\usepackage{amsfonts}
\usepackage{epic}
\usepackage{eepic}
 \usepackage{times}
\usepackage[matrix,arrow]{xy}
\usepackage{macros}
\usepackage{epsfig}

\theoremstyle{plain}

\newtheorem{thm}{Theorem}
\newtheorem{propn}{Proposition}
\newtheorem{lem}{Lemma}

\newtheorem{defn}{Definition}

\theoremstyle{remark}

\newcommand{\sst}{\scriptscriptstyle}

\renewcommand{\1}{\one}
\renewcommand{\2}{\two}

\newcommand{\beq}{\begin{equation}}
\newcommand{\eeq}{\end{equation}}

\newcommand{\pa}{\partial}
\newcommand{\ot}{\otimes}
\newcommand{\ra}{\to}

\newcommand{\SRN}{{\rm N}}

\newcommand{\de}{\delta}

\newcommand{\la}{\lambda}

\newcommand{\bz}{\bar{z}}

\newcommand{\CA}{{\mathcal A}}
\newcommand{\CB}{{\mathcal B}}

\newcommand{\CC}{{\mathcal C}}
\newcommand{\CD}{{\mathcal D}}

\newcommand{\CF}{{\mathcal F}}

\newcommand{\CL}{{\mathcal L}}
\newcommand{\CM}{{\mathcal M}}

\newcommand{\CO}{{\mathcal O}}

\newcommand{\CR}{{\mathcal R}}

\newcommand{\CW}{{\mathcal W}}

\newcommand{\SA}{{\mathsf A}}
\newcommand{\SB}{{\mathsf B}}
\newcommand{\SC}{{\mathsf C}}
\newcommand{\SD}{{\mathsf D}}

\newcommand{\SM}{{\mathsf M}}

\newcommand{\SO}{{\mathsf O}}

\newcommand{\SQ}{{\mathsf Q}}

\newcommand{\ST}{{\mathsf T}}

\newcommand{\su}{{\mathsf u}}
\newcommand{\sv}{{\mathsf v}}

\newcommand{\one}{{\mathfrak 1}}
\newcommand{\two}{{\mathfrak 2}}

\newcommand{\BC}{{\mathbb C}}

\newcommand{\BS}{{\mathbb S}}

\newcommand{\BZ}{{\mathbb Z}}

\newcommand{\srn}{{\sst\rm N}}
\newcommand{\SRM}{{\rm M}}
\newcommand{\srm}{{\sst\rm M}}

\newcommand{\rf}[1]{(\ref{#1})}
\renewcommand{\bz}{{\mathbf z}}

\newcommand{\aufz}
{\begin{list}{$\bullet$}{\topsep0cm \itemsep0cm \parsep0cm}}
\newcommand{\eaufz}{\end{list}}

\begin{document}
\title{\begin{center}Completeness of Bethe Ansatz by Sklyanin SOV for Cyclic Representations of Integrable Quantum Models \end{center}}

\vspace{2cm}

\author{\hspace{0.2cm} G. Niccoli$^{(1)}$}

\address{\hspace{0.2cm} $^{(1)}$ YITP, Stony Brook University, New York 11794-3840, USA\\ \hspace{0.8cm} email: niccoli@max2.physics.sunysb.edu }

\maketitle

{\vspace{-8cm} \tt {YITP-SB-11-1}}

\vspace{8cm}

{\bf Abstract}
\\
\begin{quotation}
\hspace{-0.78cm} In \cite{NT} an integrable quantum model was introduced and a class of its cyclic representations was proven to define lattice regularizations of the Sine-Gordon model. Here, we analyze general cyclic representations of this integrable quantum model by extending the spectrum construction introduced in \cite{GN10} in the framework of the Separation of Variables (SOV) of Sklyanin. We show that as in \cite{NT} also for general representations, the transfer matrix spectrum (eigenvalues and eigenstates) is completely characterized in terms of polynomial solutions of an associated functional Baxter equation. Moreover, we prove that the method here developed has two fundamental built-in features: i) the completeness of the set of the transfer matrix eigenstates constructed from the solutions of the associated Bethe ansatz equations, ii) the existence and complete characterization of the Baxter $\SQ$-operator.
\end{quotation}

\newpage

\tableofcontents

\newpage

\section{Introduction}
In \cite{NT} an integrable quantum model was introduced\footnote{From here on we will refer to it as SG model.} and a class of its cyclic representations was proven to define lattice regularizations of the Sine-Gordon model\footnote{This lattice regularization goes back to \cite{FST,TTF83,IK} and it is related to formulations which have more recently been studied in  \cite{FV94,BBR,Ba08}.} in a given sector of the quantum theory. Main results of that paper are
the simplicity and the completeness of the characterization of the transfer matrix spectrum. 

Here, some comments are probably in order to better point out  the relevance of these results.

{\bf Integrability}: It is of common use to call integrable a quantum model which admits a one parameter family of
commuting conserved charges. This is the case in the framework of the quantum inverse scattering method  \cite{FT79, KS79, FST} where the
transfer matrix defines the family of commuting conserved charges. Anyhow, it is worth remarking that
a more precise definition of quantum integrability requires that this family of commuting charges is a complete set of observables; i.e. its spectrum has to be simple (non-degenerate). This last statement has to be independently proven as it was done in \cite{NT} for the class of cyclic representations there analyzed.

The first fundamental task to solve for a given integrable quantum model is the exact solution of its spectral
problem, i.e. the determination of the eigenvalues and the simultaneous eigenstates of the complete set of its
commuting conserved charges.

{\bf Completeness of Bethe Ansatz}: In the framework of the Bethe ansatz several methods has been
introduced to analyze this spectral problem which lead to the characterization of the eigenvalues (and for some
of those to the eigenstates) in terms of solutions to an associated system of {\it Bethe ansatz equations}.
Hereafter, a rigorous proof of the completeness of the spectrum characterization was a fundamental and complicate long
standing goal\footnote{It was addressed also by numerical methods like the one introduced in \cite{ADMcCoy92,FMcCoy01}, see also \cite{NR03} for an application of it.} in the literature. In \cite{NT} this result has been reached showing in particular that the eigenstates of the
transfer matrix constructed from the solutions of an associated system of Bethe ansatz equations form a basis of the representation. Note that early there were only a few examples of integrable quantum models where the
completeness has been proven, including the XXX Heisenberg model; see \cite{MTV} and references therein.

There is a double motivation for the present article. On the one hand, we are interested in showing that the strong results of
simplicity and completeness of the transfer matrix spectrum can be proven for general cyclic representations
of the SG model. On the other hand, we want to use these representations to define a method to characterize
the spectrum which is based only on the Separation of Variables (SOV) of Sklyanin  \cite{Sk1,Sk2,Sk3}. Then, these
general representations are of special interest as they define concrete examples of cyclic representations of quantum integrable
models to which the {\it standard} construction\footnote{Here, we refer to the Baxter's construction by {\it gauge} transformations which leave
unchanged the transfer matrix while make triangular the action of the Lax matrices on the $\SQ$-operator. Let us recall that there are also others constructions of the $\SQ$-operator; interesting examples are presented in \cite{BLZ-I}-\cite{RW02}.} \cite{Ba73}-\cite{DKM} of the Baxter $\SQ$-operator by using cyclic dilogarithm functions does not apply in general. 
\subsection{Methodological aim}
In the framework of quantum integrability, there are several methods to analyze the spectral problem as the coordinate Bethe ansatz  \cite{Be31, Bax82, ABBQ87}, the $\ST\SQ$ method   \cite{Bax82}, the
algebraic Bethe ansatz (ABA) \cite{FT79, KS79, FST}, the analytic Bethe ansatz  \cite{R83}. However, they suffer in general from one or more of the
following problems: i)  {\it Reduced applicability}; i.e. there exist important examples of quantum integrable models to which these
methods do not apply. ii) {\it Analysis reduced only to the set of eigenvalues}; i.e. they do not allow for the construction of the
eigenstates. iii)  {\it Lack of completeness proof}; i.e. the completeness of the spectrum description is not assured by
the method while there are known cases for which these methods lead to incomplete spectrum description \cite{FK10}.

The separation of variables (SOV) method of Sklyanin is a more promising approach. Indeed, the SOV method resolves problems
like the reduced applicability of other methods; in particular, it works for integrable quantum models to which ABA does not apply.
Moreover, it leads to both the eigenvalues and the eigenstates of the transfer matrix with a construction which has as built-in
feature the completeness of the spectrum description.

In the case of cyclic representations \cite{Ta} of integrable quantum models, however, it is worth pointing out that the SOV
method  leads to the characterization of the transfer matrix spectrum in terms of solutions of an associated finite
system\footnote{The number of equations in the system is finite and related to the dimension of the cyclic representation.} of
Baxter-like equations. Thus, SOV method does not lead directly to the standard spectrum characterization expected in the
framework of Bethe ansatz approaches. Moreover, such SOV characterization
of the spectrum is not the most efficient in view of the analysis of the continuum limit. In order to solve these problems it is
important to prove that the SOV characterization can be reformulated in terms of an associated functional Baxter equation. One
possibility is to add to the SOV characterization the $\SQ$-operator approach. In particular, this can be
done when the corresponding Baxter equation is proven to be compatible with the finite system of Baxter-like equations of
the SOV characterization. In fact, this was the strategy followed in \cite{NT} for the special cyclic representations there analyzed.

The main methodological aim of the present article is to present an approach which allows to reduce the solution of the spectral
problem to the classification of the solutions of an associated Baxter equation  in a fixed class of functions. An approach where
the SOV characterization of spectrum is the starting point and  the {\it standard} construction by {\it gauge} transformations\footnote{Note that this method can be applied only when the
existence of some model dependent {\it quantum dilogarithm} functions
\cite{FK2}-\cite{BT06} is proven.} of the $\SQ$-operator is completely
bypassed. Let us recall that
such an approach has been first developed in \cite{GN10} for the special class of cyclic representations studied also in \cite{NT}. However, in \cite{GN10}
we have used some shortcuts to define the coefficients of the associated functional Baxter equations which hold only for such a
special class of cyclic representations. Here, we present
this approach in the framework of the general cyclic representations also to infer some
universal features of the method for applications to others integrable quantum
models.
\subsection{Organization of the paper}
In Section 2, it is recalled the definition of the SG quantum model in the framework of the quantum inverse scattering method. In
Section 3, the construction of the SOV representations corresponding to general cyclic representations of the SG model is
implemented. There, the result derived in \cite{NT} is extended and completed with the characterization of the coefficients of the SOV
representations. Section 4 is the
core of the article; there it is explained the method which allows to reformulate the SOV characterization of the spectrum in terms
of an associated Baxter functional equation. Moreover, there it is proven the existence of all the polynomial solutions of this
equation required to completely characterize the transfer matrix spectrum (eigenvalues and eigenstates) of the SG model. Appendix A and B contains some important properties. In Appendix C,
the comparison between our SG model and the $\tau^{(2)}$-model is made, pointing out the difference in the spectrum of the two models. In particular, it is shown as for even chains the transfer matrices are not similar for general representations.
\vspace*{6mm}
{\par {\small
{\em Acknowledgments.} I would like to thank J. Teschner for stimulating discussions on related subjects and for the opportunity to collaborate with him at DESY during the last three years under support of his grant Marie Curie Excellence Grant MEXTCT-2006-042695. I would like to thank also B. McCoy and J.-M. Maillet for the interest shown in this work. I gratefully acknowledge my current support from National Science Foundation, grants PHY-0969739.}}
\newpage

\section{Lattice SG quantum model}\label{SOV}
\subsection{Definitions and first properties}
The Lax operator of the SG model, as introduced in \cite{NT}, reads:
\begin{equation}\label{Lax}
\begin{aligned}
 L^{\rm\sst SG}_n(\la)
  &= \frac{\kappa_n}{i} \left( \begin{array}{cc}i\,\su_n^{}(q^{-\frac{1}{2}}\kappa_n^{}\sv_n^{}+q^{+\frac{1}{2}}\kappa^{-1}_n\sv_n^{-1}) &
\la_n^{} \sv_n^{} - \la^{-1}_n \sv_n^{-1}  \\
 \la_n^{} \sv_n^{-1} - \la^{-1}_n \sv_n^{} &
i\,\su_n^{-1}(q^{+\frac{1}{2}}\kappa^{-1}_n\sv_n^{}+q^{-\frac{1}{2}}\kappa_n^{}\sv_n^{-1})
 \end{array} \right) ,
\end{aligned}\end{equation}
where $\la_n\equiv\la/\xi_n$ for any $n\in \{1,...,\SRN\}$ with $\xi_n$ and  $\kappa_n$ parameters of the representation.
For any $n \in \{1,...,\SRN\}$ the couple of operators ($\su_n$,$\sv_n$) define a Weyl algebra ${\cal W}_{n}$:
\begin{equation}\label{Weyl}
\su_n\sv_m=q^{\de_{nm}}\sv_m\su_n\,,\qquad{\rm where}\;\;
q=e^{-\pi i \beta^2}\,.
\end{equation}
We will restrict our attention to the case in which $q$ is a $p$-root of unity:
\begin{equation}\label{beta}
\beta^2\,=\,\frac{p'}{p}\,,\,\,\,\,\,\, p \equiv 2l+1, p' \equiv 2l'\,\, \text{ and  }\,\, l,l'\in\BZ^{>0}\,\,\, \rightarrow \,\,\,\,q^{p}=1. 
\end{equation}
In this case each Weyl algebra  ${\cal W}_{n}$ admits a finite-dimensional representation of dimension $p$. Let us denote with:
\begin{equation}\label{u-basis}
|\, \bz\,\rangle\equiv|\, z_1,\dots,z_\SRN\,\rangle\,\,\, \text{with}\,\,\, z_i\in\BS_p\equiv\{q^{2n};n=0,\dots,2l\}\,\,\, \text{and}\,\, i\in\{1,...,\SRN\},
\end{equation}
the generic state of the basis constructed by the tensor product of the eigenstates $|\,z_n\,\rangle$ of each operator $\su_n$, then in
this basis the  Weyl algebra representation reads:
\begin{equation}\label{reprdef}
\begin{aligned}
&\sv_n\,|\, z_1,\dots,z_\SRN\rangle=\,v_n |\, z_1,\dots,q z_n,\dots,z_\SRN\rangle\,,\\
&\su_n\,|\, z_1,\dots,z_\SRN\rangle=\,u_n\,z_n |\, z_1,\dots,z_\SRN\rangle\,.
\end{aligned}
\end{equation}
where $u_n$ and $v_n$ are other parameters of the representation.

The monodromy matrix of the model is defined in terms of the Lax operators by:
\begin{equation}\label{Mdef}
\SM(\la)\,\equiv\,\left(\begin{matrix}\SA(\la) & \SB(\la)\\
\SC(\la) & \SD(\la)\end{matrix}\right)\,\equiv\,
L_\SRN^{}(\la)\dots L_1^{}(\la)\,,
\end{equation}
and it satisfies the quadratic relations:
\begin{equation}\label{YBA}
R(\la/\mu)\,(\SM(\la)\ot 1)\,(1\ot\SM(\mu))\,=\,(1\ot\SM(\mu))\,(\SM(\la)\ot 1)R(\la/\mu)\,,
\end{equation}
w.r.t. the six-vertex $R$-matrix:
\begin{equation}\label{Rlsg}
 R(\la) =
 \left( \begin{array}{cccc}
 q^{}\la-q^{-1}\la^{-1} & & & \\ [-1mm]
 & \la-\la^{-1} & q-q^{-1} & \\ [-1mm]
 & q-q^{-1} & \la-\la^{-1} & \\ [-1mm]
 & & &  q\la-q^{-1}\la^{-1}
 \end{array} \right) \,.
\end{equation}
Then the elements of $\SM(\la)$ generate a representation $\CR_\SRN$ of the so-called Yang-Baxter algebra characterized by
the $4\SRN$ parameters $\kappa=(\kappa_1,\dots,\kappa_\SRN)$, $\xi=(\xi_1,\dots,\xi_\SRN)$, $u=(u_1,\dots,u_\SRN)$ and
$v=(v_1,\dots,v_\SRN)$. In particular, the commutation relations \rf{YBA} lead to the
mutual commutativity of the elements of the one parameter family of operators:
\begin{equation}\label{Tdef}
\ST^{}(\la)\,=\,{\rm tr}_{\BC^2}^{}\SM(\la)\,,
\end{equation}
known as  transfer matrix.

In \cite{NT} the spectral problem of this transfer matrix has been solved for {\it untwisted} representations:
\begin{equation}\label{S-adj-C-untwisted}
\kappa_n^2\,\in\,\mathbb R,\,\,\,\,\xi_n^2\,\in\,\mathbb R,\,\,\,\,\,u_n^{2p}=1,\,\,v_n^{2p}=1,\,\,\,\forall \,n\,\in\,\{1,...,\SRN\}.
\end{equation}
In the present paper we will extend the analysis considering general representations:
\begin{equation}\label{S-adj-C+general}
\kappa_n^2\,\in\,\mathbb R,\,\,\,\,\xi_n^2\,\in\,\mathbb R,\,\,\,\,\,(u_n)^{*}\,u_n=1,\,(v_n)^{*}\,v_n=1,\,\,\,\forall \,n\,\in\,\{1,...,\SRN\},
\end{equation}
which we call  {\it twisted} for $u_n^{2p}\neq1$ and $v_n^{2p}\neq1$.

Let us comment that the above restrictions on the parameters of the representation are imposed to get the self-adjointness of the transfer matrix:
\begin{lem}
If the parameters of the representation satisfy the constrains \rf{S-adj-C+general} and
\begin{equation}\label{cond-T-Normality}
\varepsilon\equiv-(\kappa _{n}\xi _{n})/\left(\kappa _{n}^{\ast}\xi_{n}^{\ast}\right)\,\,\,\text{ is uniform along the chain,}
\end{equation}
then the generators of the Yang-Baxter algebra satisfy the following transformations under
Hermitian conjugation: 
\begin{equation}\label{Hermit-Monodromy}
\SM(\la)^\dagger\equiv\left( 
\begin{array}{cc}
\SA^{\dagger }(\lambda ) & \SB^{\dagger }(\lambda ) \\ 
\SC^{\dagger }(\lambda ) & \SD^{\dagger }(\lambda )%
\end{array}%
\right) =\left( 
\begin{array}{cc}
\SD(\lambda ^{\ast }) & \SC(\varepsilon\lambda ^{\ast }) \\ 
\SB(\varepsilon\lambda ^{\ast }) & \SA(\lambda ^{\ast })%
\end{array}
\right) ,
\end{equation}
which, in particular, imply the self-adjointness of the transfer matrix $\ST(\la)$ for real $\lambda $.
\end{lem}
\begin{proof}
It is simple to observe that the Lax operator of the SG model satisfies the equation \rf{Hermit-Monodromy}, which can be also
written as:
\begin{equation}\label{Hermit-L}
\left(L^{\rm\sst SG}_n(\la)\right)^{\dagger}=\sigma_1\,L^{\rm\sst SG}_n(\varepsilon\la^{*})\,\sigma_1\,\,\,\,\,\,\underset{\rf{cond-T-Normality}}{\longrightarrow}\,\,\,\,\,\,\left(\SM(\la)\right)^{\dagger}=\sigma_1\,\SM(\varepsilon\la^{*})\,\sigma_1,
\end{equation}
that is \rf{Hermit-Monodromy} holds, when one takes into account that $\SA(\la)$ and $\SD(\la)$ are even in $\la$.
\end{proof}
In the case of a lattice with $\SRN$ even quantum sites, we can introduce the operator:
\begin{equation}\label{topological-charge}
\Theta =\prod_{n=1}^{\SRN}\sv_{n}^{(-1)^{1+n}},
\end{equation}
which plays the role of a {\it grading operator} in the Yang-Baxter algebra:

\begin{lem}{\bf (Proposition 6 of \cite{NT})} \ $\Theta $ commutes with the transfer matrix and satisfies
the following commutation relations with the entries of the monodromy matrix:
\begin{eqnarray}
\Theta \SC(\lambda ) &=&q\SC(\lambda )\Theta \text{, \ \ \ }[\SA(\lambda ),\Theta
]=0, \\
\SB(\lambda )\Theta &=&q\Theta \SB(\lambda ),\text{ \ \ }[\SD(\lambda ),\Theta ]=0.
\end{eqnarray}
\end{lem}
Moreover, the $\Theta $-charge allows to express the asymptotics of the transfer matrix as:
\begin{equation}
\lim_{\log\lambda \rightarrow \mp\infty}\lambda ^{\pm\SRN}\ST (\lambda )=\left(
\prod_{a=1}^{\SRN}\frac{\kappa _{a}\xi _{a}^{\pm 1}}{i}\right) \left(
\Theta +\Theta^{-1}\right).  \label{asymptotics-t}
\end{equation}
Let us denote with $\Sigma _{\ST}$ the set of the eigenvalue functions $t(\lambda )$ of the transfer matrix $\ST(\lambda )$. By the definitions \rf{Lax} and \rf{Tdef}, $\Sigma_{\ST}$ is a subset of $\mathbb{R}[\lambda ^{2},\lambda^{-2}]_{\bar\SRN/2}$, where we are using the notations:
\begin{equation}
\bar\SRN\equiv\SRN+{\rm e}_{\SRN}-1,\,\,\,\,\,\, \text{and}\,\,\rm e_\SRN=0\,\, \text{for}\,\, \SRN\,\, \text{odd and}\,\, 1\,\, \text{for}\,\, \SRN\,\, \text{even}, 
\end{equation}
and $\mathbb{R}[x,x^{-1}]_{\SRM}$ denotes the linear space in the field $\mathbb{R}$ of the {\it real}:
\begin{equation}
f(x)\in\mathbb{R}[x,x^{-1}]_{\SRM}\,\,\, \rightarrow\,\,\, (f(x))^*=f(x^*)\,\,\, \forall x\in \mathbb{C},
\end{equation}
Laurent polynomials of degree $\SRM$ in the variable $x$.

Note that in the case of $\SRN$ even, the $\Theta $-charge naturally induces the grading $\Sigma_{\ST}=\bigcup_{k=0}^{2l}\Sigma _{\ST}^{\theta,k}$, where:
\begin{equation}
\Sigma
_{\ST}^{\theta,k}\equiv \left\{ t(\lambda )\in \Sigma _{\ST}:\lim_{\log \lambda
\rightarrow \mp \infty }\lambda ^{\pm \SRN}t(\lambda )=\left( \prod_{a=1}^{\SRN}\frac{\kappa _{a}\xi _{a}^{\pm 1}}{i}\right) (q^{k}\theta +(q^{k}\theta )^{-1})\right\} .
\end{equation}
This simply follows by the asymptotics of $\ST(\lambda )$ and by its commutativity with $\Theta $. In particular,
any $t(\lambda )\in \Sigma_{\ST}^{\theta,k}$ is a $\ST$-eigenvalue corresponding to simultaneous eigenstates of
$\ST(\lambda )$ and $\Theta $\ with $\Theta $-eigenvalue $q^{k}\theta$.

\section{Cyclic SOV representations}

\setcounter{equation}{0}

The separation of variables (SOV) of Sklyanin \cite{Sk1}-\cite{Sk3} is a method to solve the spectral problem for $\ST(\la)$ which is based on the observation that such a problem simplifies considerably if one works in a SOV representation: A representation where the commutative family of operators $\SB(\la)$ is diagonal. This is due to the simple form assumed by the operator families $\SA(\la)$ and $\SD(\la)$ in such a representation.
\subsection{Cyclic SOV representations: Generality}
In subsection \ref{SOV construction-0} we will show that the commuting family of operators\footnote{From here, we will use the index $\SRN$ when it will be need to point out that we are referring to the chain with $\SRN$ sites and we will omit it otherwise.} $\SB_\SRN(\la)$ is diagonalizable
with simple spectrum for almost all the values of the $4\SRN$ parameters $\kappa$, $\xi$, $u$ and $v$ of our model. The
corresponding SOV representations of the SG model are to a large extend determined by the Yang-Baxter
algebra \rf{YBA}; here we recall the general form of these representations.

Let $\langle\,\eta\,|$ be the generic element of a basis of eigenvectors of $\SB_\SRN(\la)$:
\begin{equation}\label{Bdef}
\langle\,\eta\,|\SB_\SRN(\la)\,=\,\eta_\SRN^{{\rm e}_\SRN}\,b_\eta(\la)\,\langle\,\eta\,|\,,\qquad b_\eta(\la)\,\equiv\,
\prod_{n=1}^{\SRN}\frac{\kappa _{n}}{i}\prod_{a=1}^{[\SRN]}\left( \la/\eta_a-\eta_a/\la\right)\,,
\end{equation}
where $[\SRN]\equiv\SRN-{\rm e}_\SRN$ and
\begin{equation}
\eta \in{\mathbb B_\SRN}\,\equiv\,\big\{\,(q^{k_1}\zeta_1,\dots,q^{k_\SRN}\zeta_\SRN)\,;\,(k_1,\dots,k_\SRN)\in\BZ_p^\SRN\,\big\}\,.
\end{equation}
Here, the simplicity of the spectrum of $\SB_\SRN(\la)$ is equivalent to the requirement $\zeta_a^p\neq\zeta_b^p$ for any $a\neq b \in \{1,\dots,[\SRN]\}$.
The remaining generators of the Yang-Baxter algebra read:
\begin{align}\label{SAdef}
\SA_\SRN(\la)\,=\,&\,{\rm e}_{\SRN}^{}\,b_\eta(\la)\left[
\frac{\la}{\eta_\SA^{}}\ST^+_\SRN-\frac{\eta_\SA^{}}{\la}\ST^-_\SRN\right]
+\sum_{a=1}^{[\SRN]}\prod_{b\neq a}\frac{\la/\eta_b-\eta_b/\la}{\eta_a/\eta_b-\eta_b/\eta_a} \,{\tt a}_\SRN^{}(\eta_a)\,\ST_a^-\,,\\
\SD_\SRN(\la)\,=\,&\,{\rm e}_{\SRN}^{}\,b_\eta(\la)\left[\frac{\la}{\eta_\SD^{}}\ST^-_\SRN-\frac{\eta_\SD^{}}{\la}\ST^+_\SRN\right]
+\sum_{a=1}^{[\SRN]}\prod_{b\neq a}\frac{\la/\eta_b-\eta_b/\la}{\eta_a/\eta_b-\eta_b/\eta_a} \,{\tt d}_\SRN^{}(\eta_a)\,\ST_a^+\,,
\label{SDdef}\end{align}
where $\ST_a^{\pm}$ are
the operators defined by
\[
\langle\,\eta_1,\dots,\eta_\SRN\,|\ST_a^{\pm}=\langle\,\eta_1,\dots,q^{\pm 1}\eta_a,\dots,\eta_\SRN\,|\,.
\]
While, $\SC_\SRN(\la)$ is uniquely\footnote{Note that the operator $\SB_\SRN(\la)$ is invertible except for $\la$ which coincides with a zero of $\SB_\SRN$, so in general $\SC_\SRN(\la)$ is defined by (4.5) just inverting $\SB_\SRN(\la)$. This is enough to fix in an unique way the operator $\SC_\SRN$ being it a Laurent polynomial of degree [$\SRN$] in $\la$.} defined by the quantum determinant relation:
\begin{equation}\label{qdetdef}
{\rm det_q}\SM(\la)\,\equiv\,
\SA(\la)\SD(q^{-1}\la)-\SB(\la)\SC(q^{-1}\la),
\end{equation}
where ${\rm det_q}\SM(\la)$ is a central element\footnote{The centrality of the quantum determinant in the Yang-Baxter algebra was first discovered in \cite{IK81}; see also \cite{IK09} for an historical note.} of the Yang-Baxter algebra \rf{YBA} which reads: \begin{equation}\label{q-det-f}
{\rm det_q}\SM(\la)\equiv\prod_{n=1}^{N}\kappa _{n}^{2}(\lambda /\mu
_{n,+}-\mu _{n,+}/\lambda )(\lambda /\mu _{n,-}-\mu _{n,-}/\lambda ),
\end{equation}
where $\mu _{n,\pm }\equiv\pm i\kappa _{n}^{\pm 1}q^{1/2}\xi _{n}$. Remark that the quantum determinant does not depend
from the parameters $u _n$ and $v _n$ of the representation and so it is the same for twisted and untwisted representations.

The expressions  \rf{SAdef} and \rf{SDdef} contain complex-valued
coefficients $\eta_\SA^{}$, $\eta_\SD^{}$, ${\tt a}_\SRN(\eta_r)$ and
${\tt d}_\SRN(\eta_r)$. Note that $Z_\SA\equiv\eta_\SA^{p}$ and $Z_\SD\equiv\eta_\SD^{p}$ are fixed by the asymptotics relations \rf{ZAD} while the coefficients ${\tt a}_\SRN(\eta_r)$ and
${\tt d}_\SRN(\eta_r)$ are restricted  by the quantum determinant condition:
\begin{equation}\label{addet}
{\rm det_q}\SM(\eta_r)\,=\,
{\tt a}_\SRN(\eta_r){\tt d}_\SRN(q^{-1}\eta_r)\,, \quad\forall r=1,\dots,[\SRN]\,.
\end{equation}
In a SOV representation, some freedom is left in the choice of ${\tt a}_\SRN(\eta_r)$ and
${\tt d}_\SRN(\eta_r)$ which can be parametrized by what we call a gauge transformation:
\begin{equation}\label{gauge}
{\tt a}_\SRN'(\eta_r)\,=\,{\tt a}_\SRN(\eta_r)\frac{f(\eta_rq^{-1})}{f(\eta_r)}\,,
\qquad
{\tt d}_\SRN'(\eta_r)\,=\,{\tt d}_\SRN(\eta_r)\frac{f(\eta_rq)}{f(\eta_r)}\,;
\end{equation}
which just amounts in a renormalization in the states of the $\SB$-eigenbasis:
\begin{equation}
\langle\,\eta\,|\,\rightarrow\,\prod_{r=1}^{\SRN}f^{-1}(\eta_r)\langle\,\eta\,|\,.
\end{equation}
\subsection{Central elements: Average values}\label{Avvalapp}
Following \cite{Ta}, let us define the average value $\CO$ of the elements of the monodromy matrix $\SM(\la)$ as
\begin{equation}\label{avdef}
\CO(\Lambda)\,=\,\prod_{k=1}^{p}\SO(q^k\la)\,,\qquad \Lambda\,=\,\la^p,
\end{equation}
where $\SO$ can be
$\SA_\SRN$, $\SB_\SRN$,
$\SC_\SRN$ or $\SD_\SRN$ and we have to remark that the commutativity of each family of operators $\SA_\SRN(\lambda )$, $\SB_\SRN(\lambda )$, $\SC_\SRN(\lambda )$ and $\SD_\SRN(\lambda )$ implies
that the corresponding average values are functions of $\Lambda$. So that $\CA_\SRN(\Lambda)$, $\CD_\SRN(\Lambda)$ are even Laurent polynomials of degree $\bar\SRN$ while $\CB_\SRN(\Lambda)$, $\CC_\SRN(\Lambda)$ are odd Laurent polynomials of degree $[\SRN]$  in $\Lambda$.

\begin{propn} \text{}   
\begin{itemize}
\item[a)] The average values $\CA_\SRN(\Lambda)$, $\CB_\SRN(\Lambda)$, $\CC_\SRN(\Lambda)$, $\CD_\SRN(\Lambda)$ of the monodromy matrix elements are central elements which satisfy the following relations:
\begin{equation}\label{H-cj-A-D}
(\CA_\SRN(\Lambda))^{*}\equiv\CD_\SRN(\Lambda^*), \ \ \ \ \ (\CB_\SRN(\Lambda))^*\equiv\CC_\SRN(\varepsilon\Lambda^*),
\end{equation}
under complex conjugation. 
\item[b)] Let $\CM(\Lambda)$ be the 2$\times$2 matrix with elements the average values  of the elements of the monodromy matrix $\SM(\la)$ , then it holds:
\begin{align}\label{RRel1a}
\CM_{\SRN}^{}(\Lambda)\,=\,
\CL_{\SRN}^{}(\Lambda)\,\CL_{\SRN-1}^{}(\Lambda)\,\dots\,\CL_1^{}(\Lambda)\,.
\end{align}
where $\CL_n(\Lambda)$ is the 2$\times$2 matrix with elements the average values of the elements of the Lax matrix $L_n^{SG}(\la)$.
\end{itemize}
\end{propn}
A similar statement was first proven in \cite{Ta}.

\textit{Proof of a).} Centrality of $\CB_\SRN(\Lambda)$ trivially follows from the fact that $\SB_\SRN(\la)$ is diagonal in the SOV representation and from \rf{Bdef} it is easily found:
\begin{equation}\label{CB}
\CB_\SRN(\Lambda)\,=\,Z_{{\SRN}}^{{\rm e}_{\SRN}}
\prod_{n=1}^{\SRN}\frac{K_n}{i^p}\prod_{a=1}^{{[}\SRN{]}}(\Lambda/Z_a-Z_a/\Lambda)\,,\qquad
\begin{aligned}
& Z_a\equiv \eta_a^p\,,\\
& K_a\equiv\kappa_a^p\,.
\end{aligned}
\end{equation}
The requirement of cyclicity of the SOV representation reads:
\begin{equation}
 ({\ST_{a}^{^-}})^p=({\ST_{a}^{^+})^p=1}\qquad \forall a\in\{1,\dots,\SRN\},
\end{equation}
so that $\CA_\SRN^{}(Z_r)$ and $\CD_\SRN^{}(Z_r)$ are centrals and related to the coefficients $a_\SRN(q^k\eta_r)$ and $d_\SRN(q^k\eta_r)$ by
\begin{equation}\label{ADaver}
\CA_\SRN^{}(Z_r)\,\equiv\,\prod_{k=1}^{p}a_\SRN(q^k\eta_r)\,,\qquad
\CD_\SRN^{}(Z_r)\,\equiv\,\prod_{k=1}^{p}d_\SRN(q^k\eta_r)\,,\qquad \forall r\in\{1,\dots,{[}\SRN{]}\}.
\end{equation}
Note that $\mathcal{A}_{\SRN}(\Lambda )\Lambda ^{\bar\SRN}$ and $\mathcal{D}_{\SRN}(\Lambda)\Lambda^{\bar\SRN}$ are polynomials in $\Lambda^{2} $ of degree $\bar\SRN$. So the
centrality of $\mathcal{A}_{\SRN}(\Lambda )$ and $\mathcal{D}_{\SRN}(\Lambda )$ follows from the simplicity of the
$\SB_\SRN$-spectrum and the centrality in the special values \rf{ADaver} to which we have to add for even $\SRN$ the centrality of the leading asymptotic terms of $\mathcal{A}_{\SRN}(\lambda )$ and $\mathcal{D}_{\SRN}(\lambda )$ as discussed in appendix \ref{Asymp-A-D}. Finally, the Hermitian conjugation properties of the elements of the monodromy matrix:
\begin{equation}
(\SA_\SRN(\lambda))^{\dag}\equiv\SD_\SRN(\lambda^*), \ \ \ \ \ (\SB_\SRN(\lambda))^\dag\equiv\SC_\SRN(\varepsilon\lambda^*),
\end{equation}
imply \rf{H-cj-A-D} and in particular the centrality of $\mathcal{C}_{\SRN}(\Lambda )$.\hspace{6.5cm}$\square$

\textit{Proof of b).} Under the assumption that $\SB(\la)$ is diagonalizable and with simple spectrum in the entire chain as well as in each subchain, the point \textit{b)} follows inductively by using the next Lemma.  \hspace{14.1cm}$\square$

\begin{lem}\label{average value} {\bf (Proposition 3 of \cite{NT})}\ 
The following recursive equations on the average values hold:
\begin{eqnarray}
\mathcal{B}_{\SRN}(\Lambda ) &=&\mathcal{A}_{\SRM}(\Lambda )\mathcal{B}%
_{\SRN-\SRM}(\Lambda )+\mathcal{B}_{\SRM}(\Lambda )\mathcal{D}_{\SRN-\SRM}(\Lambda ),
\label{average value-B} \\
\mathcal{C}_{\SRN}(\Lambda ) &=&\mathcal{D}_{\SRM}(\Lambda )\mathcal{C}%
_{\SRN-\SRM}(\Lambda )+\mathcal{C}_{\SRM}(\Lambda )\mathcal{A}_{\SRN-\SRM}(\Lambda ),
\label{average value-C} \\
\mathcal{A}_{\SRN}(\Lambda ) &=&\mathcal{A}_{\SRM}(\Lambda )\mathcal{A}%
_{\SRN-\SRM}(\Lambda )+\mathcal{B}_{\SRM}(\Lambda )\mathcal{C}_{\SRN-\SRM}(\Lambda ),
\label{average value-A} \\
\mathcal{D}_{\SRN}(\Lambda ) &=&\mathcal{D}_{\SRM}(\Lambda )\mathcal{D}%
_{\SRN-\SRM}(\Lambda )+\mathcal{C}_{\SRM}(\Lambda )\mathcal{B}_{\SRN-\SRM}(\Lambda ),
\label{average value-D}
\end{eqnarray}
where on the l.h.s. there are average values of the monodromy matrix elements on the complete chain with $\SRN$-sites while on the r.h.s. there are those of the monodromy matrices on the subchains $\1$ and $\2$ with $(\SRN-\SRM)$-sites and $\SRM$-sites, respectively.
\end{lem}
{\bf Remark 1.} \ Being the average values central elements of the representation they are unchanged by
similarity transformations and they therefore represent parameters of the representation. Moreover, the gauge
transformations clearly leave $Z_r$, $Z_\SA$, $Z_\SD$, $\CA_\SRN^{}(Z_r)$ and $\CD_\SRN^{}(Z_r)$
unchanged. Therefore these last numbers characterize gauge-invariant dates of the SOV representations.

\subsubsection{\bf{Calculation of average values}}
It is a simple exercise to show that the average values of the elements of the Lax matrices $L_n^{\rm\sst SG}(\la)$ are explicitly given by
\begin{align}\label{L_n}
\CL_n(\Lambda)&
\,=\,\frac{1}{i^p}
\left(
\begin{matrix} i^p U_n^{}(K_n^{2}V_n^{}+V_n^{-1}) &
K_n(\Lambda V_n/X_n-X_n/V_n\Lambda) \\
K_n(\Lambda/ X_nV_n-X_nV_n/\Lambda) &
i^p U_n^{-1}(K_n^{2}V_n^{-1}+V_n^{})
\end{matrix}
\right),
\end{align}
where we have used the notations
$K_n=\kappa_n^p$,  $X_n=\xi_n^p$,
$U_n=u_n^p$ and $V_n=v_n^p$. This formula together with equality \rf{RRel1a} allows to uniquely define the average values of the monodromy matrix elements and to state:
\begin{lem}\label{Fun-dep}
$\CO_\SRN^{}(\Lambda)/\prod_{n=1}^{\SRN}K_n$ for $\CO=\CA,\CB,\CC,\CD$ are Laurent polynomials of maximal degree 1 in each of the parameters $K_n,X_n,U_n,V_n$.
\end{lem}
Note that the previous lemma also implies that the gauge-invariant dates
of the SOV representations, $Z_r$, $Z_\SA$, $Z_\SD$, $\CA_\SRN^{}(Z_r)$ and $\CD_\SRN^{}(Z_r)$ (up to permutations of $r=1,\dots,[\SRN]$), are uniquely defined in terms of the parameters of the representation $K_n,X_n,U_n,V_n$.

\subsubsection{\bf{Choice of the gauge in the SOV representation}}\label{Choice-gauge}
Let us recall that the SOV coefficients ${\tt a}(\eta _{r})$ and ${\tt d}(\eta _{r})$ are specified only by the average value relations \rf{ADaver} required by the cyclicity of the SOV representation. Then we can fix the gauge in the SOV representation fixing a couple ${\tt a}(\lambda)$ and ${\tt d}(\lambda)$ of Laurent polynomial solutions of the following average relations:
\begin{equation}\label{a-d-functions}
\CA(\Lambda)+\gamma\CB(\Lambda)\,=\,\prod_{k=1}^{p}{\tt a}(q^k\lambda)\,,\qquad
\CD(\Lambda)+\delta\CB(\Lambda)\,=\,\prod_{k=1}^{p}{\tt d}(q^k\lambda)\,,
\end{equation}
where $\gamma$ and $\delta$ are constants to be fixed. As it will be clear in the next sections, it is important to require that these solutions satisfy the further conditions:
\begin{equation}\label{R1}
\left( {\tt a}(\lambda )\right) ^{\ast }={\tt d}(\lambda ^{\ast }),\hspace{6.1cm}
\end{equation}
\begin{equation}\label{R2}
\text{if \ \ }a_0\in\ \text{\large{ Z}}_{{\tt a}(\la)}\,\,\rightarrow\,\, q^k a_0\notin \text{\large{ Z}}_{{\tt a}(\la)}\,\,\,\forall k\in\{1,...,2l\},
\end{equation}
\begin{equation}\label{R3}
\text{\large{ Z}}_{{\tt a}(\la)}\cap\ \text{\large{ Z}}_{\prod_{h=0}^{2l-1}{\tt d}(\lambda q^{h})}=\emptyset,\hspace{4.3cm}
\end{equation}
where $\text{\large{ Z}}_{f(\la)}$ denotes the set of the zeros of the function $f(\la)$. Let us denote with:
\begin{equation}\label{string}
\text{\large s}_{p,\la_0}\equiv(\la_0,q\la_0,...,q^{2l}\la_0)\in\mathbb{C}^p,\,\,\,\,\,\,\,\,\, \text{a $p$-string of center}\,\, \la_0, 
\end{equation}
then condition \rf{R2}, in particular, implies that ${\tt a}(\la)$ and ${\tt d}(\la)$ are free from $p$-strings\footnote{Here and in the following we say that a function $f(\la)$ is free from $p$-strings meaning
that the set of its zeros does not contain $p$-strings, i.e. $\text{\large s}_{p,
\la_0}\not\subset\text{\large{Z}}_{f(\la)}\,
\, \forall \la_0\in\mathbb{C}.$}. Note that in the untwisted representations
($U_n=V_n=1$), we have defined these Laurent polynomials as solutions of \rf{a-d-functions} with $\gamma=-\delta=1$ which also satisfy \rf{R1}-\rf{R3}. This was possible as for untwisted representations it holds:
\begin{equation}
\CA(\Lambda)=\CD(\Lambda), \ \ \ \ \ \CB(\Lambda)=\CC(\Lambda).
\end{equation}
It is worth to note that this properties together with \rf{H-cj-A-D} imply that for untwisted representations the
average values of the elements of
the monodromy matrix are {\it real} Laurent polynomial in $\Lambda$. However, this reality condition is lost for the
twisted representations here considered; indeed, we have the following:
\begin{lem}\label{Average A-not-D} Almost for all the values of the parameters  $K_n,X_n,U_n,V_n$ the Laurent polynomials $\CA(\Lambda)$ and $\CD(\Lambda)$ are not identical as  well as $\CB(\Lambda)$ and $\CC(\Lambda)$.
\end{lem}
\begin{proof}
The Lemma is true for $\SRN=1$, as it is clear from the one-site average formula \rf{L_n}. Then, for $\SRN>1$, the proof follows by induction using Lemma $\ref{average value}$. 
\end{proof}
The previous Lemma and the Hermitian conjugation properties \rf{H-cj-A-D} imply that for generic twisted representations the only way to simultaneously satisfy \rf{a-d-functions} and \rf{R1} is by imposing $\gamma=\delta=0$. Then representing:
\begin{equation}\label{a-def}
\CA_\SRN(\Lambda)\,\equiv\,\CA_\SRN\prod_{n=1}^{\bar\SRN}(\Lambda/Z_n^{\CA}-Z_n^{\CA}/\Lambda)\,,
\end{equation}
 we can chose to define ${\tt d}(\la)$ by \rf{R1} and ${\tt a}(\la)$ as any even Laurent polynomial of the form:
\begin{equation}\label{a-def}
{\tt a}(\la)\,\equiv\,{\tt a}_\SRN\prod_{n=1}^{\bar\SRN}(\la/z_n^{\tt a}-z_n^{\tt a}/\la)\, \,\, \,\, \, \text{with}\,\, {\tt a}_\SRN^p\equiv\CA_\SRN \,\,\text{and}\,\, (z_n^{\tt a})^p\equiv Z_n^{\CA},
\end{equation}
with asymptotics for $\SRN$ even:
\begin{equation}\label{a-asymp}
\lim_{\log \lambda\rightarrow \mp \infty }\lambda ^{\pm \SRN}{\tt a}(\lambda )=\left( \prod_{a=1}^{\SRN}i\kappa _{a}\xi _{a}^{\pm 1}\right)\theta^{\pm 1}q^{(1\mp 1)\SRN},
\end{equation}
where $\theta$ is a fixed $p$-root of the average value of the $\Theta$-charge. Moreover, we impose the following prescriptions:
\begin{itemize}
\item[a)] In the case of real or imaginary zeros of $\CA(\Lambda)$:
\begin{equation}
\text{If \ }\exists\, n\in \{1,...,\bar\SRN\}:Z_{n}^{\CA}\in \text{\large{ Z}}_{\CA(\Lambda)},\,(Z_{n}^{\CA})^2\in\mathbb{R}\,\, \rightarrow\,\, z_{n}^{\tt a}\in\text{\large{ Z}}_{\tt a(\la)},\,(z_{n}^{\tt a})^2/q\in\mathbb{R}.\,\,\,\,\,\,\,\,\,\,\,\,\,\,\,
\end{equation}
\item[b)] In the case of multiple zeros of $\CA(\Lambda)$ with multiplicity $r$:
\begin{equation}
\text{If \ }\exists\, n_{1},...,n_{r}\in \{1,...,\bar\SRN\}:Z_{n_{1}}^{\CA}=...=Z_{n_{r}}^{\CA}\in \text{\large{ Z}}_{\CA(\Lambda)}\,\, \rightarrow\,\, z_{n_{1}}^{\tt a}=...=z_{n_{r}}^{\tt a}\in \text{\large{ Z}}_{\tt a(\la)}.
\end{equation}
\item[c)] In the case of couples of complex conjugate zeros of $\CA(\Lambda)$:
\begin{equation}
\text{If \ }\exists\, n\neq m\in \{1,...,\bar\SRN\}:(Z_n^{\CA})^{\ast }=Z_m^{\CA}\in \text{\large{ Z}}_{\CA(\Lambda)}\,\,\, \rightarrow\,\,\,(z_n^{\tt a}/q)^{\ast }= z_m^{\tt a}\in \text{\large{ Z}}_{\tt a(\la)}.\,\,\,\,\,\,\,\,\,\,\,\,\,\,
\end{equation}
\end{itemize}
Then under these conditions ${\tt a}(\la)$ and ${\tt d}(\la)$ satisfy the requirements \rf{R1}-\rf{R3}.  
\subsection{Constructive proof of the existence of cyclic SOV representations}\label{SOV construction-0}
In the following subsections we will show by recursive construction the

\begin{thm} \label{SOVthm}
Almost for all the values of the parameters  $K_n,X_n,U_n,V_n$ there exists a SOV representation for the SG model, in fact, we can prove that the one-parameter operator family $ \SB(\la)$ is diagonalizable with simple spectrum.
\end{thm}

\subsubsection{{\bf Recursive construction of $\SB$-eigenstates}}\label{SOV construction}
We will construct the eigenstates $\langle\,\eta\,|$ of $\SB(\la)\equiv \SB_\srn(\la)$
recursively by induction on $\SRN$.

In the case $\SRN=1$ we may
simply take $\langle\,\eta_1\,|\,=\,\langle\, v_1\,|, $ where  $\langle\, v_1\,|$ is an eigenstate of
the operator $\sv_1$ with eigenvalue $v$. It is useful to note that the inhomogeneity parameter
determines the subset of $\BC$ on which the variable $\eta_1$ lives, $\eta_1\in(\xi_1/v_1)\BS_p$.

Let us assume we have constructed the $\SB_{\srm}$-eigenbasis for any $\SRM<\SRN$, then the eigenstates
$\langle\,\eta\,|$ of $\SB_{\srn}(\la)$ may be constructed in the following form
\begin{equation}
\langle\,\eta\,|\,=\,\sum_{\chi_{\1}^{}}\sum_{\chi_\2^{}}\,K_\srn^{}(\,\eta\,|\,\chi_{\2}^{};\chi_\1^{}\,)\,
\langle \,\chi_{\2}^{}\,|\ot\langle\,\chi_\1^{}\,|\;,
\end{equation}
where $\langle \,\chi_{\2}^{}\,|$ and $\langle \,\chi_{\1}^{}\,|$ are eigenstates of $\SB_\srm(\la)$ and $\SB_{\srn-\srm}(\la)$
with eigenvalues parameterized as in \rf{Bdef}  by the tuples $\chi_\2^{}=(\chi_{\2 a}^{})_{a=1,\dots ,\srm}^{}$
and $\chi_\1^{}=(\chi_{\1 a}^{})_{a=1,\dots,\srn-\srm}^{}$, respectively. It suffices to consider the cases where $\SRN-\SRM$ is
odd.

\begin{itemize}
\item[]\hspace{-0.8cm}{\it\ref{SOV construction}.a} \  \ {\it Dependence of the kernel $K_\srn(\,\eta\,|\,{\chi}_{\2}^{};{\chi}_\1^{}\,)$ w.r.t. $\chi_{\2}$ and $\chi_{\1}$.}
\end{itemize}
From the formula
\begin{equation}
\SB_\srn(\la)=\SA_{\2 \ \srm}(\la)\SB_{\1 \ \srn-\srm}(\la)+{\SB}_{\2 \ \srm}(\la)\SD_{\1 \ \srn-\srm}(\la)
\end{equation}
it follows that the matrix elements of the kernel $K_\srn(\,\eta\,|\,{\chi}_{\2}^{};{\chi}_\1^{}\,)$ have to satisfy the relations
\begin{equation}\begin{aligned}\label{recrel}
\big(\SA_{\2 \ \srm}(\la)\SB_{\1 \ \srn-\srm}(\la) +{\SB}_{\2 \ \srm}(\la)& \SD_{\1 \ \srn-\srm}(\la)\big)^t\,K_\srn^{}(\,\eta\,|\,\chi_{\2}^{};\chi_\1^{}\,)
\\
&\,=\,\eta_\srn^{{\rm e}_\srn}\prod_{n=1}^{\SRN}\frac{\kappa _{n}}{i}\prod_{a=1}^{[\SRN]}\left( \la/\eta_a-\eta_a/\la\right)\,K_\srn^{}(\,\eta\,|\,\chi_{\2}^{};\chi_{\1}^{}\,)\,,
\end{aligned}\end{equation}
where we used the notation $\SO^t$ for the transpose of an operator $\SO$. Let us assume that
\begin{equation}
\chi_{\1a}q^{h_{1}}\notin\text{\large{Z}}_{{\rm det_q}\SM_{\1, \SRN-\SRM}(\la)}, \ \ \chi _{\2b}q^{h_{2}}\notin\text{\large{Z}}_{{\rm det_q}\SM_{\2, \SRM}(\la)}\ \ \text{and} \ \ \chi _{\1a}q^{h_{1}}\neq \chi _{\2b}q^{h_{2}},
\end{equation}
where $h_{i}\in \{1,...,p\}$, $a\in \{1,...,\SRN-\SRM\}$ and $b\in \{1,...,\SRM\}$. Under these assumptions the previous equations yield recursion relations for
the dependence of the kernel in the variables $\chi_{\1 a}$ and $\chi_{\2 b}$ simply by setting $\la=\chi_{\1 a}$
and $\la=\chi_{\2 b}$. Indeed for $\la=\chi_{\1 a}$
the first term on the left of \rf{recrel} vanishes leading to
\begin{equation}\label{recrel1}\begin{aligned}
 {\ST_{\1 a}^{^-}
 K_\srn^{}(\,\eta\,|\,\chi_{\2}^{};\chi_{\1}^{}\,)}\;&{d_\1^{}(q^{-1}\chi_{\1a}^{})}\; \,\chi_\srm^{{\rm e}_\srm}\prod_{n=1}^{\SRN-\SRM}
\frac{i}{\kappa _{n}}\prod_{a=1}^{[\rm M]}(\chi_{\1 a}^{}/\chi_{\2 b}^{}-\chi_{\2 b}^{}/\chi_{\1 a}^{})\, \\[-1ex]
& \,=\, {K_\srn^{}(\,\eta\,|\,\chi_{\2}^{};\chi_{\1}^{}\,)}\; \,\eta_\srn^{{\rm e}_\srn}\prod_{b=1}^{[\SRN]}(\chi_{\1 a}^{}/\eta_b^{}-\eta_b^{}/\chi_{\1 a}^{})
\,,
\end{aligned}\end{equation}
while for $\la=\chi_{\2a}$ one finds similarly
\begin{equation}\label{recrel2}\begin{aligned}
{\ST_{\2 a}^{^+}
 K_\srn^{}(\,\eta\,|\,\chi_{\2}^{};\chi_{\1}^{}\,)}\; &{a_\2^{}(q^{+1}\chi_{\2a}^{})}\;
 \prod_{n=1}^{\rm M}\frac{i}{\kappa _{n}} \prod_{b=1}^{\SRN-\rm M}(\chi_{\2 a}^{}/\chi_{\1 b}^{}-\chi_{\1 b}^{}/\chi_{\2 a}^{} )\,
\\[-1ex] &\,=\,{K_\srn^{}(\,\eta\,|\,\chi_{\2}^{};\eta_{\1}^{}\,)}\; \,\eta_\srn^{{\rm e}_\srn}\prod_{b=1}^{[\SRN]}
(\chi_{\2 a}^{}/\eta_b^{}-\eta_b^{}/\chi_{\2 a}^{})
\,,
\end{aligned}\end{equation}
where ${d_\1^{}(\chi_{\1 a}^{})}$ and ${a_\2^{}(\chi_{\2 a}^{})}$ are the known coefficients of the SOV representations in the subchains $\1$ and $\2$.  
If ${\rm M}$ is even we find the recursion relation determining the dependence on $\chi_{\2\srm}$ by sending $\la\ra\infty$
in \rf{recrel}, leading to
\begin{equation}\label{recrelzero}
{\ST_{\2 \srm}^{^+}
 K_\srn^{}(\,\eta\,|\,\chi_{\2}^{};\chi_{\1}^{}\,)}\; \frac{1}{\chi_{\2 \SA}^{}}\,\prod_{a=1}^{\rm M-1}\frac{1}{\chi_{\2 a}^{}}\,\prod_{b=1}^{\SRN-\rm M}\frac{1}{\chi_{\1 b}^{}}\,
\,=\,{K_\srn^{}(\,\eta\,|\,\chi_{\2}^{};\eta_{\1}^{}\,)}\;
\prod_{b=1}^{\SRN}\frac{1}{\eta_b^{}}
\,.
\end{equation}
\newpage 

\begin{itemize}
\item[]\hspace{-0.8cm}{\it\ref{SOV construction}.b} \  \  {\it Dependence of the kernel $K_\srn(\,\eta\,|\,{\chi}_{\2}^{};{\chi}_\1^{}\,)$ w.r.t. $\eta$.}
\end{itemize}
The action
\begin{equation}
\langle \eta |\SD_\srn(\eta _{i})=-{\rm det_q}\SM_{\2 \ \srm}(q\eta _{i})\langle \eta
|{\SB}_{\2 \ \srm}^{-1}(q\eta _{i}){\SB}_{\1 \ \srn-\srm}(\eta _{i}),
\end{equation}
follows from the formula
\begin{equation}\label{D-exp}
\SD_\srn(\la)=\SD_{\2 \ \srm}(\la)\SD_{\1 \ \srn-\srm}(\la)+{\SC}_{\2 \ \srm}(\la)\SB_{\1 \ \srn-\srm}(\la)
\end{equation}
when we express ${\SC}_{\2 \ \srm}(\la)$ by the quantum determinant ${\rm det_q}\SM_{\2 \ \srm}(\la)$ in the subchain $\2$ and we use that
\begin{equation}
\langle \eta |{\SB}_{\1 \ \srn-\srm}(\eta _{i})=-\langle \eta |{\SD}_{\1 \ \srn-\srm}(\eta _{i}){\SA}_{\2 \ \srm}^{-1}(q\eta _{i}){\SB}_{\2 \ \srm}(q\eta _{i}).
\end{equation}
So the kernel matrix elements $K_\srn(\,\eta\,|\,{\chi}_{\2}^{};{\chi}_\1^{}\,)$ have to satisfy the following recursion relations for
the dependence w.r.t. $\eta$
\begin{equation}\label{recrel+1}\begin{aligned}
 {d_\srn(\eta_{i}^{})}{\ST_{i}^{^+}
 K_\srn^{}(\,\eta\,|\,\chi_{\2}^{};\chi_{\1}^{}\,)}=-\frac{{K_\srn^{}(\,\eta\,|\,\chi_{\2}^{};\chi_{\1}^{}\,)}{\rm det_q}\SM_{\2 \ \srm}(q\eta _{i})}{\chi_{\2 \srm}^{{\rm e}_\srm}\prod_{a=\SRN-\SRM+1}^{\SRN}\kappa_{a}/i}\frac{\prod_{a=1}^{\SRN-\SRM}(\kappa_{a}/i)\left( \eta _{i}/\chi_{\1 a}-\chi_{\1 a}/\eta _{i}\right) }{\prod_{a=1}^{[\SRM]}\left(q\eta _{i}/\chi _{\2 a}-\chi _{\2 a}/q\eta _{i}\right) }
\,.
\end{aligned}\end{equation}
If ${\rm N}$ is even we have to consider further the asymptotic of the equation \rf{D-exp} which by the asymptotic behavior, analyzed in appendix \ref{Asymp-A-D}, leads to the identity:
\begin{equation}
\lim_{\log \lambda \rightarrow \mp \infty }\lambda ^{\pm \SRN}\SD_\srn(\la)=-\lim_{\log \lambda \rightarrow \mp \infty }\lambda ^{\pm \SRN}{\rm det_q}\SM_{\2 \ \srm}(\la q){\SB}_{\2 \ \srm}^{-1}(\la q){\SB}_{\1 \ \srn-\srm}(\la).
\end{equation}
From which we get the following recursion relation determining the kernel dependence on $\eta_{\srn}$:
\begin{equation}\label{recrel+zero}
{\ST_{\srn}^{^+} K_\srn^{}(\,\eta\,|\,\chi_{\2}^{};\chi_{\1}^{}\,)\; \eta_{\SD}^{}}\,\prod_{a=1}^{\rm N-1}\eta_{a}^{}\,=\,{K_\srn^{}(\,\eta\,|\,\chi_{\2}^{};\eta_{\1}^{}\,)}\;
\prod_{b=1}^{\SRM}\frac{\xi_{\srn-\srm+b}^{2}}{\chi_{\2 \ b}^{}}\,\prod_{b=1}^{\SRN-\SRM}\chi_{\1 \ b}^{}\,.
\end{equation}

\subsubsection{\bf{Determination of gauge-invariant SOV dates: $Z_r,\,Z_{\SA},\,Z_{\SD},\,\CA_\SRN^{}(Z_r),\,\CD_\SRN^{}(Z_r)$}}\label{gauge-invariant-dates}
The condition of compatibility among the recursion relations \rf{recrel1}-\rf{recrel2} and \rf{recrelzero} and the requirement of cyclicity,  $ ({\ST_{\1
a}^{^-}})^p=({\ST_{\2 a}^{^+})^p=1}$, provide a system of $\SRN$ algebraic equations in the $\SRN$ unknown
$Z_a\equiv\eta_a^p$:
\begin{equation}\label{Exrecrel1}
\CD_{\1 \,\SRN-\SRM}^{}(\chi_{\1a}^{p})\; (\chi_{\2\srm}^{{\rm e}_\srm})^{p}\prod_{n=1}^{\SRN-\SRM}
\frac{i^p}{\kappa _{n}^p}\prod_{b=1}^{[{\rm M}]}
 (\chi_{\1 a}^{p}/\chi_{\2 b}^{p}-\chi_{\2 b}^{p}/\chi_{\1 a}^{p})\,
\,=\, (\eta_\srn^{{\rm e}_\srn})^{p}\prod_{b=1}^{[\SRN]}(\chi_{\1 a}^{p}/\eta_b^{p}-\eta_b^{p}/\chi_{\1 a}^{p})\,,
\end{equation}
\begin{equation}\label{Exrecrel2}
\CA_{\2 \,\SRM}^{}(\chi_{\2 a}^p)\; \prod_{n=1}^{\SRM}
\frac{i^p}{\kappa _{n}^p}
 \prod_{b=1}^{\SRN-\rm M}(\chi_{\2 a}^{p}/\chi_{\1 b}^{p}-\chi_{\1 b}^{p}/\chi_{\2 a}^{p} )\,=\,
 (\eta_\srn^{{\rm e}_\srn})^{p}\prod_{b=1}^{[\SRN]}(\chi_{\2 a}^{p}/\eta_b^{p}-\eta_b^{p}/\chi_{\2 a}^{p})\,,\,\,\,\,\,\,\,\,\,\,\,\,\,\,\,\,
\end{equation}
with for $\SRM$ even 
\begin{equation}\label{Exrecrel3}
\frac{1}{\chi_{\2 \SA}^{p}}\,\prod_{a=1}^{\rm M-1}\frac{1}{\chi_{\2 a}^{p}}\,
\prod_{b=1}^{\SRN-\rm M}\frac{1}{\chi_{\1 b}^{p}}\,
\,=\,
\prod_{b=1}^{\SRN}\frac{1}{\eta_b^{p}}
\,.
\end{equation}
The simplicity of the spectrum of $\SB(\la)$ in the SOV representations of both the subchains $\1$ and $\2$ implies that the above system of equations completely determines the unknown $Z_a$ in terms of $\chi_{\2 a}^p$, $\chi_{\1 a}^p$. Indeed, we can reformulate this system of equations in terms of the following Laurent polynomial equation:
\begin{equation}\label{CB'}
\mathcal{A}_{\SRM}(\Lambda )\mathcal{B}%
_{\SRN-\SRM}(\Lambda )+\mathcal{B}_{\SRM}(\Lambda )\mathcal{D}_{\SRN-\SRM}(\Lambda )=
Z_{{\SRN}}^{{\rm e}_{\SRN}}
\prod_{n=1}^{\SRN}\frac{K_n}{i^p}
\prod_{a=1}^{{[}\SRN{]}}(\Lambda/Z_a-Z_a/\Lambda),
\end{equation}
where the l.h.s. is formed out of the known average values of the monodromy matrix elements in the subchains $\1$ and $\2$. Therefore the problem to determine the unknown $Z_a$ for $a\in \{1,...,[\SRN]\}$ is reduced to the problem to determine the zeros of the known Laurent polynomial at the l.h.s. of \rf{CB'}.

The requirement of cyclicity  $({\ST_{i}^{^+})^p=1}$ determines the remaining gauge-invariant SOV dates:
\begin{equation}\label{Exrecrel+1}
\CD_{\SRN}^{}(Z_{i})=\, -{\rm det_q}\CM_{\SRM}(Z_{i})\frac{\CB_{\SRN-\SRM}(Z_{i})}{\CB_{\SRM}(Z_{i})} , \ \ 
\CA_{\SRN}^{}(Z_{i})=\, -{\rm det_q}\CM_{\SRN-\SRM}(Z_{i})\frac{\CB_{\SRM}(Z_{i})}{\CB_{\SRN-\SRM}(Z_{i})},
\end{equation}
for any $i\in\{1,...,[\SRN]\}$ where we have defined:
\begin{equation}
{\rm det_q}\CM_{X}(\Lambda)\equiv\prod_{a=1}^{p}{\rm det_q}\SM_{X}(q^{a}\la), \ \ X={\SRM,\SRN-\SRM, \SRN},
\end{equation}
while for $\SRN$ even the condition $({\ST_{\SRN}^{^+})^p=1}$ reads:
\begin{equation}\label{Exrecrel+3}
Z_{\SD} = \prod_{a=1}^{\rm N-1}\frac{1}{Z_{a}}\,\prod_{b=1}^{\SRN-\rm M}\chi_{\1 b}^{p}\,\prod_{b=1}^{\SRN}\frac{X_{\SRN-\SRM+b}^2}{\chi_{\2 b}^{p}}\,.
\end{equation}
Note that from the identity:
\begin{equation}
{\rm det_q}\CM_{X}(\Lambda)=\CA_{X}(\Lambda)\CD_{X}(\Lambda)-\CB_{X}(\Lambda)\CC_{X}(\Lambda), \ \ X={\SRM,\SRN-\SRM},
\end{equation}
the equations \rf{Exrecrel+1} and \rf{Exrecrel+3} are equivalent to the equations \rf{average value-A} and \rf{average value-D} as well as the equation \rf{CB'} coincides with \rf{average value-B} of Lemma \ref{average value}.

\subsubsection{{\bf Completeness of $\SB$-eigenstates and simplicity of $\SB$-spectrum}}\label{nondegapp}
Here we show that the set of $\SB$-eigenstates $\langle\,\eta\,|$ constructed in the previous subsection is complete, i.e. it
defines a basis of the representation; this is done showing that there are $p^\SRN$ distinct corresponding $\SB$-eigenvalues. In particular, we can prove:
\begin{propn}
\label{B-simplicity}The SOV dates $Z_{r}$ with $r\in \{1,...,[\SRN]\}$ are all distinct for almost all the values of the parameters
$K_n,X_n,U_n,V_n$ of the SG model.
\end{propn}
\begin{proof}Note that we can apply Lemma \ref{Fun-dep} to the average values on the l.h.s. of \rf{CB'} and so we have that
the r.h.s. of \rf{CB'} times $\prod_{n=1}^{\SRN}K_n^{-1}$ is a Laurent polynomial of maximal degree 1 in each of the
parameters $K_n,X_n,U_n,V_n$. The consequent functional dependence\footnote{\label{zeros-func-dependence}Let $\sigma_n^{{[\SRN]}}(Z)$ be the degree
$n$ elementary symmetric polynomial in the variables $Z_r$, then $\sigma_n^{{[\SRN]}}(Z)/\sigma_{[\SRN]}^{[\SRN]}(Z)$ are
Laurent polynomials of maximal degree 1 in each one of the parameters
$K_n,X_n,U_n,V_n$.} of the $Z_1^{},\dots,Z_{[\SRN]}^{}$ w.r.t. these parameters implies
that it is sufficient to show the non-vanishing of the Jacobian:
\begin{equation}\label{J-nonzero}
J(X;K,U,V)\,\equiv\,
{\rm det}\left(\frac{\pa Z_r}{\pa X_s}\right)_{r,s=1,\dots,[\SRN]}\neq\,0\,
\end{equation}
for some special values of the parameters $K_n,U_n,V_n$ in order to prove that $J(X;K,U,V)\neq 0$ for almost all the values.
Whenever $J(X;K,U,V)\neq 0$, we have invertibility of the map $Z=Z(X_1,\dots,X_{[\SRN]})$ from which the claim of the Proposition follows.

To show that \rf{J-nonzero} is indeed satisfied, let us choose $K_{n}/V_{n}=\pm i^p$ for     $n=1,...,[\SRN]$, then the average values \rf{L_n} of the Lax operators simplify to
\begin{equation}
\mathcal{L}_{n}(\Lambda )=\pm\left(
\begin{array}{cc}
0 & V_{n}(\Lambda V_{n}/X_{n}-X_{n}/V_{n}\Lambda) \\
V_{n}(\Lambda/ V_{n}X_{n}-X_{n}V_{n}/\Lambda)  & 0%
\end{array}
\right) \,.
\end{equation}
Using now \rf{RRel1a} to compute the l.h.s. of \rf{CB'} we get:
\begin{equation}\label{CBsimp}
[U_\SRN^{}(K_\SRN^{2}V_\SRN^{}+V_\SRN^{-1})]^{{\rm e}_\SRN}
\prod_{n=1}^{[\SRN]}(\Lambda V_n^{(-1)^{n}}/X_{n}-X_{n}/\Lambda V_n^{(-1)^{n}} )=[Z_{{\SRN}}K_{{\SRN}}/i^p]^{{\rm e}_{\SRN}}
\prod_{a=1}^{{[}\SRN{]}}(\Lambda/Z_a-Z_a/\Lambda),
\end{equation}
from which the fact that $J(X;K,U,V)\neq 0$ trivially follows.
\end{proof}

We conclude that our construction
of $\SB(\la)$-eigenstates will work if the representations
$\CR_{\SRN}$, $\CR_{\SRM}$ and $\CR_{\SRN-\SRM}$
are all non-degenerate. Theorem 1 follows by induction.

\section{Characterization of $\ST$-spectrum}\label{Compatib}

\setcounter{equation}{0}

\subsection{SOV characterization of $\ST$-spectrum}

In the SOV representations the spectral problem for $\ST(\la)$ is reduced to the following discrete system of Baxter-like equations in the wave-function $\Psi_t(\eta)\equiv\langle\,\eta\,|\,t\,\rangle$ of a $\ST$-eigenstate $|\,t\,\rangle$:
\begin{equation}\label{SOVBax1}
t(\eta _{r})\Psi(\eta)\,=\,{\tt a}(\eta _{r})\Psi(\ST_r^-(\eta))+{\tt d}(\eta _{r})\Psi(\ST_r^+(\eta))\, \qquad \text{ \ }\forall r\in \{1,...,[\SRN]\},
\end{equation}
where $(\eta _{1},...,\eta _{\SRN})\in \mathbb{B}_{\SRN}$ with $\mathbb{B}_{\SRN}$ the set of zeros of the $\SB$-operator in the SOV representation. Here we have denoted with $\ST_r^{\pm}(\eta)\equiv(\eta_1,\dots,q^{\pm 1}\eta_r,\dots,\eta_\SRN)$ and ${\tt a}(\eta _{r})$ and ${\tt d}(\eta _{r})$ the coefficients of the SOV representation as
defined in subsection \ref{Choice-gauge}. In the case of
$\SRN$ even, we have to add to the system \rf{SOVBax1} the following equation in the variable $\eta_\SRN$:
\begin{equation}\label{SOVBax2}
\ST_\SRN^+\Psi_{k}(\eta)\,=\,q^{-k}\Psi_{k}(\eta),
\end{equation}
for $t(\lambda )\in \Sigma _{\ST}^{\theta,k}\ \ $with $k\in \{0,...,2l\}$ and $\theta\equiv\eta _{A}\prod_{a=1}^{\SRN-1}\eta
_{a}/\prod_{a=1}^{\SRN}\xi_{a}$, obtained from the asymptotics of $\ST(\la)$ given in
appendix \ref{Asymp-A-D}.
\subsection{Characterization of $\ST$-eigenvalues as solutions of a functional equation}\label{Eigenvalues-T}
Let us introduce the one parameter family $D(\la)$ of $p\times p$ matrix:
\begin{equation}\label{D-matrix}
D(\la) \equiv
\begin{pmatrix}
t(\la)   &-{\tt d}(\la)&   0        &\cdots & 0 & -{\tt a}(\la)\\
-{\tt a}(q\la)& t(q\la)&-{\tt d}(q\la)& 0     &\cdots & 0 \\
      0       & {\quad} \ddots      & &     &     &         \vdots   \\
  \vdots           &     &  \cdots    &  &       &     \vdots   \\
     \vdots         &     &   & \cdots &       &   \vdots     \\
     \vdots   &            &    &  &  \ddots{\qquad}     &   0 \\
 0&\ldots&0& -{\tt a}(q^{2l-1}\la)& t(q^{2l-1}\la) &
-{\tt d}(q^{2l-1}\la)\\
-{\tt d}(q^{2l}\la)   & 0      &\ldots      &     0  & -{\tt a}(q^{2l}\la)& t(q^{2l}\la)
\end{pmatrix}
\end{equation}
where for now $t(\la )$ is just a real even Laurent polynomial of degree $\bar\SRN$ in $\la$.
\begin{lem}
\label{detD}The determinant of the matrix $D(\la)$ is an even Laurent
polynomial of maximal degree $\bar\SRN$ in $\Lambda \equiv \lambda ^{p}$.
\end{lem}
\begin{proof}
Let us start observing that $D$$(\lambda q)$ is obtained by $D(\lambda )$ exchanging the first and\ $p$-th column and after the first and
\ $p$-th row, so that 
\begin{equation}
\det_{p}\text{$D$}(\lambda q)=\det_{p}\text{$D$}(\lambda )%
\text{ \ \ }\forall \lambda \in \mathbb{C},
\end{equation}
Being ${\tt a}(\lambda )$, ${\tt d}(\lambda )$ and $t(\lambda )$ even Laurent polynomial of degree
$\bar\SRN$ in $\lambda$ then the parity of $\det_{p}\text{$D$}(\Lambda)$ and the fact that it is a
Laurent polynomial of maximal degree $\bar\SRN$ trivially follow.
\end{proof}
The interest toward the function $\det_{p}D(\Lambda )$ is due to the following: 
\begin{lem}
\label{Charact-Sigma1}Let $t(\la)\in\Sigma _{\ST}$, then $t(\la)$ is a solution of the functional equation:
\begin{equation}
\det_{p}\text{$D$}(\Lambda)\equiv0.
\end{equation}
\end{lem}

\begin{proof} Note that from $t(\lambda )\in\Sigma _{\ST}$ and the SOV characterization of the $\ST$-spectrum we have that the requirement that the system of equations \rf{SOVBax1} admits a
non-zero solution reads:
\begin{equation}
\det_{p}\text{$D$}(\eta^{p} _{a})=0\text{ \ \ }\forall a\in \{1,...,[\SRN]\} \  \ and \  \ (\eta _{1},...,\eta _{[\SRN]})\in 
\mathbb{B}_{\SRN},  \label{compatibility}
\end{equation}
In the case of $\SRN$ even, we have moreover:
\begin{equation}
\lim_{\log \Lambda \rightarrow \mp \infty }\Lambda ^{\pm \SRN}\det_{p}\text{%
$D$}(\Lambda )=0,  \label{asymp-compatibility}
\end{equation}
which simply follows by observing that:
\begin{align}
\lim_{\log \Lambda \rightarrow \mp \infty }\Lambda ^{\pm \SRN}\det_{p}\text{%
$D$}(\Lambda )& =\det_{p}\left\Vert q^{(1\mp 1)\SRN}\theta^\pm\delta
_{i,j-1}+q^{-(1\mp 1)\SRN}\theta^\mp\delta _{i,j+1}-(q^{k}\theta +q^{-k}\theta^{-1}))\delta
_{i,j}\right\Vert   \notag \\
& \times (-1)\prod_{n=1}^{\SRN}\left( i\kappa _{n}\xi _{n}^{\pm }\right) ^{p}\left.
=\right. 0.
\end{align}
for $t(\lambda )\in\Sigma _{\ST}^{\theta,k}$ and $k\in\{0,...,2l\}$.

Then the function $\det_{p}$$D$$(\Lambda )$ is zero in $\SRN+${\rm e}$_{\SRN}$ different values of $\Lambda ^{2}$ which thanks to Lemma \ref{detD} implies the statement of the Lemma.
\end{proof}

{\bf Remark 2.}  \ Let us note that the same kind of functional equation $\det D(\Lambda )=0$ also appears in \cite{BR89, Ne02, Ne03,BBP90}. There it recasts, for different integrable quantum models at the {\it roots of unit}, the functional relations which result from the truncated fusions of transfer matrix eigenvalues; this is  in particular true for the $\tau_2$-model\footnote{The SOV representations of this model were analyzed in a series of works \cite{GIPS06, GIPST07, GIPST08}.} \cite{BS,BBP90,Ba89, Ba04}.
\subsection{Construction of Baxter equation solutions from $\ST$-eigenvalues}\label{Bax-FE}

\begin{thm}\label{Derivation-Baxter-functional}
For any $t(\lambda )\in \Sigma _{\ST}$, we can construct uniquely up to
normalization a real polynomial:
\begin{equation}\label{Q_t-definition}
Q_{t}(\lambda )=\lambda^{a_{t}}\prod_{h=1}^{2l\bar\SRN-b_{t}}(\lambda^{2}-\lambda
_{h}^{2}),\,\,\,\,\,\,\,\, 0\leq a_{t}\leq 2l,\,\,0\leq b_{t}\leq 2l\bar\SRN,
\end{equation}
which is a solution of the Baxter functional equation:
\begin{equation}\label{tq-Baxter}
t(\lambda )Q_{t}(\lambda )={\tt a}(\lambda )Q_{t}(\lambda q^{-1})+{\tt d}(\lambda
)Q_{t}(\lambda q)\ \ \ \ \forall \lambda \in \mathbb{C},\hspace{0.7cm}
\end{equation}
free from $p$-strings.
Moreover, for $\SRN$ even and $\, t(\lambda)\in \Sigma _{\ST}^{\theta,k}$ with $k\in \{0,...,2l\}$, it holds:
\begin{equation}\label{Q-even}
a_{t}=k,\,\,\,\,\,\,\,\,\, b_{t}=k\,\, \mathsf{mod}\,p.\hspace{1cm}
\end{equation}
\end{thm}

\begin{proof} {\it Construction:} \ Let us notice that the condition $t(\lambda )\in \Sigma _{\ST}$ implies that the $p\times p$ matrix $D(\lambda )$
has rank $2l$ for any $\lambda \in \mathbb{C}\backslash \{0\}$. Let us denote with
\begin{equation}\label{cofactor-def}
\textsc{C}_{i,j}(\lambda)=(-1)^{i+j}\det_{2l}D_{i,j}(\lambda)
\end{equation}
the $(i,j)$ {\it cofactor} of the matrix $D$$(\lambda)$; then the matrix formed out of these cofactors has rank $1$, i.e. all the vectors:
\begin{equation}
\text{\textsc{V}}_{i}(\lambda )\equiv (\text{\textsc{C}}_{i,1}(\lambda ),%
\text{\textsc{C}}_{i,2}(\lambda ),...,\text{\textsc{C}}_{i,2l+1}(\lambda
))^{\ST}\in \mathbb{C}^{p}\text{ \ \ }\forall i\in \{1,...,2l+1\}
\end{equation}
are proportional: 
\begin{equation}
\text{\textsc{V}}_{i}(\lambda )/\text{\textsc{C}}_{i,1}(\lambda )=\text{%
\textsc{V}}_{j}(\lambda )/\text{\textsc{C}}_{j,1}(\lambda )\text{\ \ \ \ }%
\forall i,j\in \{1,...,2l+1\},\text{ }\forall \lambda \in \mathbb{C}.
\label{covector-proport}
\end{equation} 
The proportionality (\ref{covector-proport}) of the eigenvectors \textsc{V}$_{i}(\lambda
)$ implies: 
\begin{equation}\label{proportionality}
\text{\textsc{C}}_{2,2}(\lambda )/\text{\textsc{C}}_{2,1}(\lambda )=\text{%
\textsc{C}}_{1,2}(\lambda )/\text{\textsc{C}}_{1,1}(\lambda )
\end{equation}
which, by using the property (\ref{cofactors-diagonal}), can be rewritten as:%
\begin{equation}
\label{Inter-step}
\text{\textsc{C}}_{1,1}(\lambda q)/\text{\textsc{C}}_{1,2l+1}(\lambda q)=%
\text{\textsc{C}}_{1,2}(\lambda )/\text{\textsc{C}}_{1,1}(\lambda ).
\end{equation}
Moreover, the first element in the vectorial condition\ 
$D(\lambda )$\textsc{V}$_{1}(\lambda )=$\b{0} reads: 
\begin{equation}\label{Bax-eq}
t(\lambda )\text{\textsc{C}}_{1,1}(\lambda )={\tt a}(\lambda )\text{\textsc{C}}%
_{1,2l+1}(\lambda )+{\tt d}(\lambda )\text{\textsc{C}}_{1,2}(\lambda ).
\end{equation}
In appendix \ref{Co-F Properties}, we have analyzed the properties of these cofactors and thanks to
Lemma \ref{Ap1} and \ref{cofactors-zeros} we can introduce now the polynomials
\textsc{c}$_{1,1}\overline{\text{\textsc{C}}}_{1,1}(\lambda )$,\textsc{c}$_{1,2}^*\overline{\text{\textsc{C}}}_{1,2l+1}(\lambda )$ and
\textsc{c}$_{1,2}\overline{\text{\textsc{C}}}_{1,2}(\lambda )$ defined by
simplifying the common factors in \textsc{C}$_{1,1}(\lambda )$, \textsc{C}$_{1,2l+1}(\lambda )$ and
\textsc{C}$_{1,2}(\lambda )$, respectively. In particular, from Lemma \ref{Ap1} and \ref{cofactors-zeros}
we have that $\overline{\text{\textsc{C}}}_{1,1}(\lambda )$ is an even polynomial of the form:
\begin{equation}
\overline{\text{\textsc{C}}}_{1,1}(\lambda )=\prod_{h=1}^{\bar{\SRN}_{1,1}}(\lambda^{2}-\lambda
_{h}^{2}),\,\,\,\,\,\,\,\, \bar{\SRN}_{1,1}\leq 2l\bar\SRN,
\end{equation}
which furthermore satisfies the properties:
\begin{equation}
\textit{Reality}:\,\,\,(\overline{\textsc{C}}_{1,1}(\lambda ))^{\ast }\equiv\overline{\textsc{C}}_{1,1}(\lambda ^{\ast }) \,\,\,\,\,\text{and }\,\,\,\,0\notin\text{\large{Z}}_{\overline{\textsc{C}}_{1,1}(\lambda )},\,\,\, \text{\large s}_{p,\la_0}\not\subset\text{\large{Z}}_{\overline{\textsc{C}}_{1,1}(\lambda )},\,\, \forall \la_0\in\mathbb{C}.
\end{equation} 
Now, being by definition $\text{\large{Z}}_{\overline{\textsc{C}}_{1,1}(\lambda )}\cap\text{\large{Z}}_{\overline{\textsc{C}}_{1,2}(\lambda )}=\text{\large{Z}}_{\overline{\textsc{C}}_{1,1}(\lambda )}\cap\text{\large{Z}}_{\overline{\textsc{C}}_{1,2l+1}(\lambda )}=\emptyset$, equation \rf{Inter-step} implies:
\begin{equation}\label{proportionality-cofactor}
\overline{\text{\textsc{C}}}_{1,2l+1}(\lambda )=q^{2\bar{\SRN}_{1,1}}\text{$\overline{\text{%
\textsc{C}}}$}_{1,1}(\lambda q^{-1}),\text{ \ \ $\overline{\text{\textsc{C}}}
$}_{1,2}(\lambda )=q^{-2\bar{\SRN}_{1,1}}\text{$\overline{\text{\textsc{C}}}$}_{1,1}(\lambda q)\,,\,\,\,\,\, \varphi\equiv\textsc{c}_{1,1}/\textsc{c}_{1,2}\text{ \ a phase},
\end{equation}
and then equation (\ref{Bax-eq}) assumes the form of a Baxter equation in the polynomial $\overline{\text{\textsc{C}}}_{1,1}(\lambda )$:
\begin{equation}\label{deform-BAX}
t(\lambda )\text{$\overline{\text{\textsc{C}}}$}_{1,1}(\lambda )=\bar{{\tt a}}%
(\lambda )\text{$\overline{\text{\textsc{C}}}$}_{1,1}(\lambda q^{-1})+\bar{{\tt d}}%
(\lambda )\text{$\overline{\text{\textsc{C}}}$}_{1,1}(\lambda q),
\end{equation}
with coefficients $\bar{{\tt a}}(\lambda )\equiv(q^{2\bar{\SRN}_{1,1}}\varphi){\tt a}(\lambda )$ and $\bar{{\tt d}}%
(\lambda )\equiv (q^{2\bar{\SRN}_{1,1}}\varphi)^{-1}{\tt d}(\lambda )$.

Let us denote with $\bar{D}(\Lambda )$ the matrix defined as in \rf{D-matrix} but with coefficients $\bar{{\tt a}}(\lambda )$ and $\bar{{\tt d}}(\lambda )$. Then the consistence condition $\det_{p} \bar{D}(\Lambda )\equiv0$ of the Baxter equation \rf{deform-BAX} and the condition $t(\lambda )\in \Sigma _{\ST}$ implies:
\begin{equation}\label{detDbar}
\det_{p} \bar{D}(\Lambda )-\det_{p}D(\Lambda )\equiv0\,\,\,\,\,\,\longleftrightarrow\,\,\,\,\,\,(\varphi ^{p}-1)\left( {\CA}(\Lambda)-\varphi ^{-p}{\CD}(\Lambda)\right)\equiv0,
\end{equation}
which is true if and only if $\varphi $ is a $p$-root of the unity. Note that \rf{detDbar} is derived by using the following expansion for $\det_{p}D(\Lambda )$ and a similar one for $\det_{p}\bar D(\Lambda )$:\begin{eqnarray}\label{detD-exp}
\det_{p}D(\Lambda )&=&{\CA}(\Lambda)+{\CD}(\Lambda)-{\tt a}(\lambda ){\tt d}(\lambda/q)\det_{2l-1}D_{(1,2l+1),(1,2l+1)}(\lambda )  \notag \\
&&-{\tt a}(\lambda q){\tt d}(\lambda)\det_{2l-1}D_{(1,2),(1,2)}(\lambda )+t(\lambda )\det_{2l}D_{1,1}(\lambda )\text{,}
\end{eqnarray}
where $D_{(h,k),(h,k)}(\lambda )$ denotes the $(2l-1)\times (2l-1)$
sub-matrix of $D(\lambda )$ obtained removing the rows and columns $h$ and $k$, plus the formulae\footnote{They follow from Lemma 3 of \cite{GN10} thanks to the \textit{tridiagonality} of these matrices.}:
\begin{align}
& \left. \det_{2l}\overline{D}_{1,1}(\lambda )=\det_{2l}%
D_{1,1}(\lambda ),\right.  \\
& \left. \det_{2l-1}\overline{D}_{(1,2),(1,2)}(\lambda
)=\det_{2l-1}D_{(1,2),(1,2)}(\lambda ),\right.  \\
& \left. \det_{2l-1}\overline{D}_{(1,2l+1),(1,2l+1)}(\lambda
)=\det_{2l-1}D_{(1,2l+1),(1,2l+1)}(\lambda ).\right.
\end{align}
Finally, we can define our polynomial solution \rf{Q_t-definition} of the Baxter equation \rf{tq-Baxter} by:
\begin{equation}
Q_{t}(\lambda )\equiv \lambda ^{a_{t}}\text{$\overline{\text{\textsc{C}}}$}_{1,1}(\lambda ),
\end{equation}
where $q^{-a_{t}}\equiv q^{2\bar{\SRN}_{1,1}}\varphi$ with $a_{t}\in\{0,..,2l\}$ and $b_{t}\equiv2l\bar\SRN-\bar{\SRN}_{1,1}$.

In the case $\SRN$ even, for $t(\lambda)\in \Sigma _{\ST}^{\theta,k}$ with $k\in \{0,...,2l\}$,  we get
from the Baxter equation and the asymptotics of the coefficients \rf{a-asymp} the following asymptotic
conditions:
\begin{equation}
\lim_{\lambda \rightarrow 0}\frac{Q_{t}(\lambda q)}{Q_{t}(\lambda )}=q^{k},\text{ \ \ }\lim_{\lambda \rightarrow \infty }\frac{Q_{t}(\lambda q)}{Q_{t}(\lambda )}=q^{-(2\SRN +k)}.  \label{cond-asymp+k}
\end{equation}
Then the characterization \rf{Q-even} is just a corollary of \rf{cond-asymp+k}.
\end{proof}
\begin{proof} {\it Uniqueness:} Let $t(\lambda)\in\Sigma _{\ST}$ and let $Q_{t}(\lambda )$ be the solution of the Baxter equation \rf{tq-Baxter} constructed in the previous part of the proof. Then denoting with $\bar Q_{t}(\lambda )\in\mathbb{C}[\lambda]$ any other solution, we can define the q-Wronskian:
\begin{equation}\label{q-W}
W_t(\la)\,=\,Q_t(\la)\bar Q_t(q^{-1}\la)-\bar Q_t(\la)Q_t(q^{-1}\la)\,.
\end{equation}
which by the Baxter equation satisfies the equation:
\begin{equation}
{\tt a}(\lambda )\,W_t(\lambda )\,=\,{\tt d}(\lambda )\,\ST^{+}W_t(\lambda )\,.
\end{equation}
Thanks to the cyclicity $(\ST^{+})^{p}=1$ the average of the above equation reads:
\begin{equation}
({\CA}(\Lambda )-{\CD}(\Lambda ))\,\CW_t(\Lambda)\,=0\,\,\,\,\,\text{ with }\,\,\,\,\,\CW_t(\Lambda)\equiv\prod_{k=0}^{2l}W_t(\lambda q^k),
\end{equation}which by Lemma \ref{Average A-not-D} implies $W(\lambda )\equiv 0$.
It is then easy to see that this implies that $\bar Q_t(\la)\equiv Q_t(\la)$ up to normalization when we have assumed $\bar Q_t(\la)$ free of $p$-strings.
\end{proof}

{\bf Remark 3.} \ The previous theorem implies that for any $t(\lambda )\in \Sigma _{\ST}$ the polynomial solution $Q_{t}(\lambda )$ of the Baxter equation can be related to the determinant of a tridiagonal matrix of finite size $p-1$ for the rational
$\beta^2=p'/p$. It is then relevant to remark that others determinant representations for Baxter equation solutions are known in
literature; for example in the quantum periodic Toda chain. There, these solutions are expressed in terms of linear combinations of
determinants of semi-infinite tridiagonal matrices \cite{GP,GM,KL99}. The above remark explains the interest in a careful analysis  for the SG model
of the limit $\beta^2\rightarrow\bar{\beta}^2$ with $\bar{\beta}^2$ irrational. Indeed, under this limit (i.e. $p',\ p\rightarrow + \infty$) the dimension of the representation in any quantum site diverges as well as the size of the tridiagonal matrix
associated to a solution $Q_{t}(\lambda )$ of the Baxter equation. Then, it is natural to investigate if characterizations like those
encountered for the quantum periodic Toda chain apply to the SG model for irrational  $\bar{\beta}^2$, too. In particular, this is interesting for a potential
reformulation\footnote{The use of NLIE to reformulate the spectrum characterization in integrable quantum models goes back to
\cite{KP91,KBP91}. NLIE reformulations similar to that presented in \cite{KT10} also appear in \cite{BLZ-I, Za, BT06, T}.} of the SG spectrum in terms of solutions of nonlinear integral equations (NLIE) as shown recently in \cite{KT10} for the
quantum Toda chain. On the other hand, it is worth remarking that the SG spectrum admits another NLIE reformulation which is of
DDV-type\footnote{ This type of NLIE are derived by a reformulation of the Bethe ansatz equations and they were introduced and
analyzed for fermionic lattice regularizations of the Sine-Gordon model in \cite{DDV92}-\cite{R01} and in
\cite{FR02I}-\cite{FR03II} for a related model.} and satisfies the feature of completeness as consequence of the Bethe ansatz
completeness proven in Subsection \ref{Compl-BA}. 
\subsection{Construction of $\ST$-eigenstates from $\ST$-eigenvalues}\label{T-eigenstates}
The results of the previous subsections allow both to introduce a complete characterization of the set $\Sigma_\ST$ and to construct one $\ST$-eigenstate $|t\rangle$ for any $t(\la)\in\Sigma_\ST$; indeed, we can prove the following:
\begin{thm}
The set $\Sigma _{\ST}$ coincides with the set of all the $t(\lambda )\in \mathbb{R}[\lambda ^{2},\lambda ^{-2}]_{\bar\SRN/2}$ solutions of the functional
equation:
\begin{equation}\label{I-Functional-eq}
\det_{p}\text{$D$}(\Lambda)=0,\text{ \ \ }\forall \Lambda\in \mathbb{C}.
\end{equation}
Moreover, for any $t(\la)\in\Sigma _{\ST}$ a corresponding $\ST$-eigenstate is characterized by:
\begin{equation}\label{Qeigenstate-odd}
\Psi_t(\eta)\equiv\langle\,\eta _{1},...,\eta _{\SRN}\,|\,t\,\rangle=\prod_{r=1}^{\SRN}Q_t(\eta_r),\,\,\,\,\,\,\,\,\,\,\,\,\,\,  if \, \, \SRN\,\, is\, \, odd,\hspace{3.92cm} 
\end{equation}
\begin{equation}\label{Qeigenstate-even}
\Psi_t(\eta)\equiv\langle\,\eta _{1},...,\eta _{\SRN}\,|\,t\,\rangle=\eta_\SRN^{-k}\prod_{r=1}^{\SRN-1}Q_t(\eta_r),\,\,\,\,\,\,\,\,\,\,\,\,\,\,  if \, \, \SRN\,\, is\, \,even\, \, and\, \, t(\la)\in\Sigma_\ST^{\theta,k}, 
\end{equation}
where $Q_{t}(\la)$ is the Baxter equation solution constructed in Theorem \ref{Derivation-Baxter-functional}.
\end{thm}

\begin{proof}
The inclusion of $\Sigma _{\ST}$ in that set of solutions of \rf{I-Functional-eq} was proven in Lemma
\ref{Charact-Sigma1}. To prove the reverse inclusion of sets, we have just to observe that the construction of
$Q_{t}(\la)$, presented in Theorem \ref{Derivation-Baxter-functional}, holds for any $t(\lambda )\in
\mathbb{R}[\lambda ^{2},\lambda ^{-2}]_{\bar\SRN/2}$ solution of \rf{I-Functional-eq}. Then, the $t(\la)$ and the $\Psi _{t}(\eta )$, \rf{Qeigenstate-odd} and \rf{Qeigenstate-even}, are
solutions of the system of Baxter like equations \rf{SOVBax1} and \rf{SOVBax2} and so they define a $\ST$-eigenvalue and a corresponding $\ST$-eigenstate.
\end{proof}

\subsection{Simplicity of $\ST$-spectrum}\label{Simple-T}
In this section we show that the spectrum of the transfer matrix $\ST(\lambda )$ is non-degenerate (or simple). Let us start proving the following:
\begin{lem}\label{AD-average-Z}
For almost all the values of the parameters $K_n,X_n,U_n,V_n$, the average values of the monodromy matrix elements $\SA(\la)$ and  $\SD(\la)$ satisfy the inequalities:
\begin{equation}
\CA(Z_{a})\neq \CD(Z_{a}),\,\,\,\,\,\,\,\,\,\,\,\, \forall a\in \{1,...,[\SRN]\},
\end{equation}
where $Z_{a}$ are the zeros of the average value of $\SB(\la)$.
\end{lem}
\begin{proof}
We have to prove that the functions:
\begin{equation}
\CF_a(K_n,X_n,U_n,V_n)\equiv\CA_\SRN(Z_{a})-\CD_\SRN(Z_{a}),\,\,\,\,\,\, \forall a\in \{1,...,[\SRN]\}
\end{equation}
are nonzero for almost all the values of the parameters $K_n,X_n,U_n,V_n$. The Lemma \ref{Fun-dep} and
the functional dependence\footnote{See footnote \ref{zeros-func-dependence}.} of the $Z_1^{},\dots,Z_{[\SRN]}^{}$ w.r.t. these parameters implies
that it is sufficient to show that the functions $\CF_a$ are nonzero for some special value of the parameters in order to prove our Lemma.
Note that the following identities hold:
\begin{eqnarray}\label{special-choice}
\CF_a\mid_{U_n=1\,\, \forall\,n\in\{2,...,\SRN\},\,V_n=1\,\, \forall\,n\in\{1,...,\SRN\}}=(1+K_1^2)(U_1-U_1^{-1})\CA_{\2,\,\SRN-1}(Z_{a}),
\end{eqnarray}
Here, we have used the decomposition of the chain in a first subchain $\1$, formed by the site 1, and a second subchain $\2$, formed by the remaining sites. Moreover, we have used the identities:
\begin{eqnarray}
(\CA_{\1,\,1}-\CD_{\1,\,1})&=&(1+K_1^2)(U_1-U_1^{-1}),\,\,\quad\quad\quad\CB_{\1,\,1}(\Lambda)=\CC_{\1,\,1}(\Lambda), \\ 
\CB_{\2,\,\SRN-1}(\Lambda)&=&\CC_{\2,\,\SRN-1}(\Lambda),\,\,\,\quad\quad\quad\quad\quad\quad\CD_{\2,\,\SRN-1}(\Lambda)=\CA_{\2,\,\SRN-1}(\Lambda),
\end{eqnarray}
which follows by using Lemma \ref{average value} and \rf{L_n} for the special choice of the parameters done in \rf{special-choice}. Then for $K_1\neq\pm i$  and $U_1\neq\pm 1$, we have only to show that we can always take $\CA_{\2,\,\SRN-1}(Z_{a})\neq0$ for any $a\in\{1,...,[\SRN]\}$.

Note that from the identities:
\begin{eqnarray}
{\rm det_q}\CM_{\2,\,\SRN-1}(\Lambda)&=&(\CA_{\2,\,\SRN-1}(\Lambda))^2-(\CB_{\2,\,\SRN-1}(\Lambda))^2,\,\,\,\,\\
\CB_\SRN(\Lambda)&=&\CA_{\1,\,1}\CB_{\2,\,\SRN-1}(\Lambda)+\CB_{\1,\,1}(\Lambda)\CA_{\2,\,\SRN-1}(\Lambda),
\end{eqnarray}
we have that $\CA_{\2,\,\SRN-1}(Z_{a})=0$ if and only if $Z_{a}$ is a double zero of ${\rm det_q}\CM_{\2,\,\SRN-1}
(\Lambda)$. However, this is not the case in general because the function ${\rm det_q}\CM_{\2,\,\SRN-1}
(\Lambda)$ has not double zeros for general values of the parameters $X_n$ and $K_n$; as it is clear from the formula:
\begin{equation}
{\rm det_q}\CM_{\2,\,\SRN-1}(\Lambda)\equiv\prod_{h=\pm 1}\prod_{n=2}^{N}K_{n}^{2}(\Lambda /M_{n,h}-M_{n,h}/\Lambda ), \quad M_{n,\pm }\equiv\pm i^p K_{n}^{\pm 1}X_{n}.
\end{equation}
which is obtained averaging the quantum determinant \rf{q-det-f}.
\end{proof}
\begin{thm}
For almost all the values of the parameters $K_n,X_n,U_n,V_n$ of a twisted representation, the spectrum of
$\ST(\lambda )$ is simple.
\end{thm}

\begin{proof}
We have to prove that up to normalization for any given $t(\la)\in\Sigma_\ST$ the wave-function \rf{Qeigenstate-odd} and
\rf{Qeigenstate-even} are the only solutions of the system \rf{SOVBax1} and \rf{SOVBax2}.

Let us denote with $\bar\Psi_{t}(\eta)$, $\eta\in\mathbb{B}_{\SRN}$, any other solution corresponding to the same  $\ST$-eigenvalue $t(\la)$. Then, we can define the q-Wronskian:
\begin{equation}\label{q-W}
W_{t,r}(\eta)\,=\,\Psi_{t}(\eta)\,\bar\Psi_{t}(\ST_r^-(\eta))-\bar\Psi_{t}(\eta)\,\Psi_{t}(\ST_r^-(\eta)),\,\,\,\,\forall r\in\{1,...,[\SRN]\}.
\end{equation}
which by the system of Baxter like equations \rf{SOVBax1} satisfy the equations:
\begin{equation}
{\tt a}(\eta_r )\,W_{t,r}(\eta)\,=\,{\tt d}(\eta_r )\,W_{t,r}(\ST_r^{+}(\eta)),\,\,\,\,\forall r\in\{1,...,[\SRN]\}.
\end{equation}
Thanks to the cyclicity $(\ST_r^{+})^{p}=1$, the averages of the above equations read:
\begin{equation}
({\CA}(Z_r)-{\CD}(Z_r))\,\CW_{t,r}(\eta)\,=0\,\,\,\,\,\text{ with }\,\,\,\,\,\CW_{t,r}(\eta)\equiv\prod_{k=0}^{2l}W_{t,r}(\ST_r^{k}(\eta)),
\end{equation}
which by Lemma \ref{AD-average-Z} implies:
\begin{equation}
W_{t,r}(\eta)= 0\,\,\, \leftrightarrow\,\,\,\frac{\bar\Psi_{t}(\ST_r^+(\eta))}{\bar\Psi_{t}(\eta)}=\frac{Q_{t}(\eta_r q)}{Q_{t}(\eta_r)},\,\,\,\,\,\,\,\,\forall \eta\in\mathbb{B}_{\SRN}\,\,\,\text{and}\,\,\,\forall r\in\{1,...,[\SRN]\}.
\end{equation}
In the case $\SRN$ odd, these identities just imply the uniqueness of the solution of the system
\rf{SOVBax1} and its factorized form
given in \rf{Qeigenstate-odd}. In the case $\SRN$ even, they imply the factorization:
\begin{equation}\label{Qeigens-even}
\bar\Psi_t(\eta)=\bar f_t(\eta_\SRN)\prod_{r=1}^{\SRN-1}Q_t(\eta_r).
\end{equation}
Then we get $\bar f_{t}(\eta _{\SRN})\propto \eta _{\SRN}^{-k}$ as
the unique solution of the equation \rf{SOVBax2} in $\eta _{\SRN}$ for $t(\la)\in\Sigma_\ST^{\theta,k}$.
\end{proof}
{\bf Remark 4.} \  It is worth to point out that for the twisted representations of the SG model the spectrum of the transfer matrix is simple both for $\SRN$ odd and even. This is due to the fact that by
definition for $\SRN$ even these representations are characterized by an average value of the $\Theta$-charge different from 1, i.e. $\theta^p\neq1$.
\subsection{Completeness of the Bethe ansatz}\label{Compl-BA}
An important consequence of the previous analysis is that it naturally leads to the complete characterization of the transfer
matrix spectrum (eigenvalue and eigensate) in terms of real polynomial solutions of the associated Baxter functional equation.

Let us observe that Theorem \ref{Derivation-Baxter-functional} implies that to any $t(\lambda )\in \Sigma
_{\ST}$ it is associated one and only one self-complex-conjugate solution\footnote{It means that the tuple $(\pm\la_1,\dots,\pm\la_{2l\bar\SRN-b})$ is invariant up to reordering under complex conjugation.} $(\pm\la_1,\dots,\pm\la_{2l\bar\SRN-b})$ of the following system of Bethe equations:
\begin{equation}\label{eq-Bethe}
\frac{{\tt a}(\pm\la_k)}{{\tt d}(\pm\la_k)}=-q^{2a}\prod_{h=1}^{2l\bar\SRN-b}\frac{(q^2\la_{k}^{2}-\la_{h}^{2})}{(\la_{k}^{2}/q^2-\la_{h}^{2})},\,\,\,\,\,\,\forall  k\in\{1,...,2l\bar\SRN-b\}.
\end{equation}
Conversely, let us consider a self-complex-conjugate solution of the Bethe system of equations \rf{eq-Bethe}, then it defines uniquely a real polynomial $Q(\la)$ by the equation \rf{Q_t-definition}. Now, by using $Q(\la)$ we can construct the function:
\begin{equation}\label{TfromQ}
t(\la)\,=\,({\tt a}(\la)Q(q^{-1}\la)+{\tt d}(\la)Q(q\la))/Q(\la)\,.
\end{equation}
which, thanks to the Bethe equations \rf{eq-Bethe}, is nonsingular for $\la=\pm\la_k,\,\,\forall k\in\{1,
\dots,2l\bar\SRN-b\}$ and is a real even Laurent polynomial of degree $\bar\SRN$ in $\la$.  Moreover, we
can also  uniquely construct a state $|\,t\,\rangle$ by inserting $Q(\la)$ in equation \rf{Qeigenstate-odd}
for $\SRN$ odd (and \rf{Qeigenstate-even} for $\SRN$ even). Then, $t(\la)$ and the wave-function $\Psi _{t}(\eta
)\equiv\langle\eta _{1},...,\eta _{\SRN}\,|\,t\,\rangle$ satisfies by definition the system of Baxter like
equations \rf{SOVBax1} (and
\rf{SOVBax2} for $\SRN$ even). This implies that starting from the given solution of the Bethe equation we
have uniquely
reconstructed the $\ST$-eigenvalue $t(\la)$ and the corresponding $\ST$-eigenstate $|\,t\,\rangle$.
This  establishes a one-to-one correspondence between the above solutions
to \rf{eq-Bethe} and the spectrum of the transfer matrix
({\it Completeness of the Bethe ansatz}).

\subsection{Reconstruction of $\SQ$-operator}\label{Q-op-def}
In the previous sections, we have given a complete and operative characterization of the spectrum of the
transfer matrix $\ST(\la)$ starting from its SOV characterization. It is worth pointing out
that we have derived our results without any need to introduce a $\SQ$-operator. Anyhow, it is worth noticing that our
construction has as built-in property the existence and characterization of the $\SQ$-operator.
This feature is even more relevant in the case of twisted representations of the SG model where
a direct construction by cyclic dilogarithm functions can be applied only on the sub-variety of representations generated by points on algebraic curves.  

\begin{defn}
Let $\mathsf{Q}(\lambda )$ be the operator family defined by\footnote{Here, we have decided to define $\bar{Q}_{t}(\lambda )$ by multiplying the function $Q_{t}(\lambda )$ for the $p$-string $\la^p$ when $a_{t}$ is odd. This is done to get that $\bar{Q}_{t}(\lambda )$  is always a real {\bf even} polynomial in $\la$ which is still a solution of \rf{tq-Baxter}.}: 
\begin{equation}
\mathsf{Q}(\lambda )|t\rangle \equiv\bar{Q}_{t}(\lambda )|t\rangle, \text{ \ \ where \ \ } \bar{Q}_{t}(\lambda )\equiv \la^{p(1-\rm e_{a_{t}})}Q_{t}(\lambda ),\,\,\,\,\,\forall t(\lambda
)\in \Sigma _{\ST}.
\end{equation}
Here, $Q_{t}(\lambda )$ is the real polynomial corresponding to $t(\lambda )$ by the injection defined in
Theorem \ref{Derivation-Baxter-functional} and $|t\rangle$ is the corresponding $\ST$-eigenstate.
\end{defn}

Then the following theorem holds:

\begin{thm}
The operator family $\mathsf{Q}(\lambda )$ is a Baxter $\SQ$-operator:
\begin{description}
\item[(A)] $\mathsf{Q}(\lambda )$ and $\ST(\lambda )$ satisfy the commutation
relations: 
\begin{equation}
[\mathsf{Q}(\lambda ),\ST(\mu )]=[\mathsf{Q}(\lambda ),\mathsf{Q}(\mu
)]=0\text{ \ \ }\,\,\,\,\,\,\,\,\,\,\,\,\,\,\,\,\,\,\,\forall \lambda ,\mu \in \mathbb{C}\text{,}
\end{equation}
plus the Baxter equation: 
\begin{equation}
\ST(\lambda )\mathsf{Q}(\lambda )={\tt a}(\lambda )\mathsf{Q}(\lambda
q^{-1})+{\tt d}(\lambda )\mathsf{Q}(\lambda q)\text{ \ \ }\,\,\,\,\,\forall \lambda \in 
\mathbb{C},
\end{equation}
and for $\SRN$ even: 
\begin{equation}
[\mathsf{Q}(\lambda ),\Theta]=0\text{ \ \ }\,\,\,\,\,\text{ \ \ }\,\,\,\,\,\forall \lambda \in 
\mathbb{C}.
\end{equation}
\item[(B)] $\mathsf{Q}(\lambda )$ is a polynomial of degree $2l(\bar\SRN+1)$ in $\lambda^2$: 
\begin{equation}
\mathsf{Q}(\lambda )\equiv \sum_{n=0}^{2l(\bar\SRN+1)}\mathsf{Q}_{n}\lambda ^{2n},
\end{equation}
with coefficients $\mathsf{Q}_{n}$ self-adjoint operators.
\end{description}
\end{thm}

\begin{proof}
Note that the self-adjointness of the transfer matrix $\ST(\lambda )$ implies
that $\mathsf{Q}(\lambda )$ is well defined being its action defined on
a basis. The commutation relations and the Baxter equation are trivial consequences of Definition 1, thanks to the simplicity of the $\ST$-spectrum. Finally, the reality condition of the $\SQ$-eigenvalues:
\begin{equation}
\left( \bar{Q}_{t}(\lambda )\right) ^{\ast }=\bar{Q}_{t}(\lambda ^{\ast })\text{ \ }%
\forall \lambda \in \mathbb{C}
\end{equation}
implies that the operators $\mathsf{Q}_{n}$ are self-adjoint.
\end{proof}
{\bf Remark 5.} In \cite{Ba04},  Baxter has extended the construction of the $\SQ$-operator by gauge transformations of \cite{BS} to the $\tau_2$-model for general cyclic representations not restricted to those parameterized by points on algebraic curves. The main tool there was a generalization of the discrete dilogarithm functions. In particular, he has remarked that asking the cyclicity only for the products of couples of these functions the $\SQ$-operator can be still well defined for general cyclic representations of the model. Note that there are no doubts that this construction can be adapted to our general representations of the SG model leading to the representation of the $\SQ$-operator in the basis \rf{reprdef} of the representation. However, it is also clear that this last characterization of $\SQ$ cannot add any information to our construction of the transfer matrix spectrum but it can be potentially useful for more  {\it physical aims}. In particular, we can try to use it to investigate the lattice dynamics \cite{FV,BKP93} associated to the general cyclic representations of the SG model. This can be done following the the same line drawn in \cite{NT}, i.e. showing that the discrete evolution operator can be written in terms of the $\SQ$-operator. Here, the main {\it physical aim} is to prove that these representations of the SG model define lattice regularizations of the Sine-Gordon model for general $\alpha$-sectors of the quantum theory (see Section 2.2 of \cite{NT} for the definition).
\appendix
\section{Properties of the cofactors \textsc{C}$_{i,j}(\protect\lambda )$}\label{Co-F Properties}

\setcounter{equation}{0}

Let $t(\la )$ be just a real even Laurent polynomial of degree $\bar\SRN$ in $\la$, then the cofactors of the matrix $D(\la )$ satisfy the following:
\begin{lem}
\label{Ap1} The cofactors $\text{\textsc{C}}_{i,j}(\lambda )$ are even Laurent polynomials of maximal degree\footnote{The $a_{i,j}$ and $b_{i,j}$ are
nonnegative integers and by definition $\la_{(i,j),h}\neq 0$ for any $h\in \{1,...,2l\bar\SRN-(a_{i,j}+b_{i,j})\}$.}
$2l\bar\SRN$:
\begin{eqnarray}\label{general-cofactor}
\text{\textsc{C}}_{i,j}(\lambda ) &=&\text{\textsc{c}}_{i,j}\lambda
^{-2l\bar\SRN+2a_{i,j}}\prod_{h=1}^{2l\bar\SRN-(a_{i,j}+b_{i,j})}(\lambda^2-\lambda
_{{(i,j)},h}^2),
\end{eqnarray}
which satisfy the following properties: 
\begin{equation}
\text{\textsc{C}}_{h+i,k+i}(\lambda )=\text{\textsc{C}}_{h,k}(\lambda q^{i})%
\text{ \ \ \ }\forall i,h,k\in \{1,...,p\}.
\label{cofactors-diagonal}
\end{equation}
and
\begin{equation}\label{cofactor-cc}
\left( \text{\textsc{C}}_{1,1}(\lambda )\right) ^{\ast }\equiv \text{\textsc{%
C}}_{1,1}(\lambda ^{\ast }),\text{ \ \ \ }\left( \text{\textsc{C}}%
_{1,2}(\lambda )\right) ^{\ast }\equiv \text{\textsc{C}}_{1,2l+1}(\lambda
^{\ast }).
\end{equation}
\end{lem}

\begin{proof}
The characterization \rf{general-cofactor} of the cofactors trivially follows from being ${\tt a}(\lambda )$, ${\tt d}(\lambda )$ and $t(\lambda )$ even Laurent polynomial of degree $\bar\SRN$ in $\lambda$. Note that by the definition \rf{cofactor-def} of the cofactors $\textsc{C}_{i,j}(\lambda)$ the equations (\ref{cofactors-diagonal}) are simple consequences of $q^{p}=1$\ and are proven exchanging rows and columns in the determinants.

Now let us prove the property \rf{cofactor-cc} for the cofactor \textsc{C}$_{1,1}(\lambda )=\det_{2l}$$D_{1,1}(\lambda )$, where:
\begin{equation}
D_{1,1}(\lambda )\equiv\left\Vert t(\lambda
q^{h})\delta _{h,k}-{\tt a}(\lambda q^{h})\delta _{h,k+1}-{\tt d}(\lambda
q^{h})\delta _{h,k-1}\right\Vert _{1\leq h\leq2l,1\leq k\leq2l},  \label{D11}
\end{equation}
then by the properties under complex conjugation of $t(\la),\,{\tt a}(\la)$ and ${\tt d}(\la)$ it holds:
\begin{equation}
(D_{1,1}(\lambda ))^{*}\equiv\left\Vert t(\lambda^{*}
q^{h})\delta _{h,k}-{\tt d}(\lambda^{*} q^{p-h})\delta _{h,k+1}-{\tt a}(\lambda^*
q^{p-h})\delta _{h,k-1}\right\Vert _{1\leq h\leq2l,1\leq k\leq2l},  \label{D11*}
\end{equation}
Let us denote with $D_{1,1}^\CC$ the $2l\times2l$ matrix of columns:
\begin{equation}
\, C_a^{D_{1,1}^\CC}\equiv\,C_{p-a}^{D_{1,1}},\,\,\,\,\,\,\,\,\,\forall a\in\{1,...,2l\},
\end{equation}
where $C_a^{X}$ stays for the column $a$ of the matrix $X,$ and similarly with $D_{1,1}^{\CC,\CR}$ the $2l\times2l$ matrix of rows:
\begin{equation}
R_a^{D_{1,1}^{\CC,\CR}}\equiv\,R_{p-a}^{D_{1,1}^{\CC}},\,\,\,\,\,\,\,\,\,\forall a\in\{1,...,2l\},
\end{equation}
where $R_a^{X}$ stays for the row $a$ of the matrix $X$. Then we get the identity:
\begin{equation}
D_{1,1}^{\CC,\CR}(\la^{*})\equiv
(D_{1,1}(\lambda ))^{*}\,\,\,\rightarrow\,\,\,
\left( \text{\textsc{C}}_{1,1}(\lambda )\right) ^{\ast }\equiv \text{\textsc{%
C}}_{1,1}(\lambda ^{\ast }),
\end{equation}
being the determinant invariant under an even number of column and row exchanges. In a completely similar way, we can prove the identity:
\begin{eqnarray}\label{D_(1,2)-prop}
(\det_{2l-1}D_{(1,2),(1,2)}(\la ))^*\equiv\det_{2l-1}D_{(1,2),(1,2)}(\la^*/q ),
\end{eqnarray}
then by using  the cofactor expansions:
\begin{eqnarray}
\text{\textsc{C}}_{1,2l+1}(\lambda ) &=&\prod_{h=1}^{2l}{\tt a}(\lambda
q^{h})+{\tt d}(\lambda/q )\det_{2l-1}D_{(1,2),(1,2)}(\lambda/q ),
\label{C_2,1-expan} \\
\text{\textsc{C}}_{1,2}(\lambda ) &=&\prod_{h=1}^{2l}{\tt d}(\lambda
q^{h})+{\tt a}(\lambda q)\det_{2l-1}D_{(1,2),(1,2)}(\lambda )\text{%
.}  \label{C_1,2-expan}
\end{eqnarray}
the second identity in \rf{cofactor-cc} is a corollary of \rf{D_(1,2)-prop}.
\end{proof}
\begin{lem}\label{cofactors-zeros}
Let $t(\lambda )\in \Sigma _{\ST}$, then the following identities hold:
\begin{equation}\label{co-degree}
a_{1,1}=a_{1,2},\,\,\,\,\,\,\,\,\,\,b_{1,1}=b_{1,2}\hspace{4.2cm}
\end{equation}
where $a_{i,j}$ and  $b_{i,j}$ are the nonnegative integers defined in \rf{general-cofactor},
\begin{equation}\label{Co-F prop2}
\text{\large{Z}}_{\text{\textsc{C}}_{1,1}(\la)}\cap\text{\large{Z}}_{\text{\textsc{C}}_{1,2}(\la)}\equiv\text{\large{Z}}_{\text{\textsc{C}}_{1,1}(\la)}\cap
\text{\large{Z}}_{\text{\textsc{C}}_{1,2l+1}(\la)},\hspace{1.5cm}
\end{equation}
and if
\begin{equation}\label{Co-F string}
\exists\,\, \text{\large s}_{p,\la_0}\subset\text{\large{Z}}_{\text{\textsc{C}}_{1,1}(\la)}\,\,\,\rightarrow\,\,\, \text{\large s}_{p,\la_0}\cap\text{\large{Z}}_{\text{\textsc{C}}_{1,2}(\la)}\neq\emptyset,\hspace{1.3cm}
\end{equation}
for $\text{\large s}_{p,\la_0}\equiv(\la_0,q\la_0,...,q^{2l}\la_0)$ any $p$-string.
\end{lem}
\begin{proof}
Under the assumption $t(\lambda )\in \Sigma _{\ST}$ we have proven the validity of the equations
\rf{Inter-step} and \rf{Bax-eq}. Then the identity \rf{co-degree} follows by using \rf{Inter-step}  and the
property \rf{cofactor-cc}. Concerning the proof of identity \rf{Co-F prop2} thanks to the prescriptions
\rf{R1}-\rf{R3}, chosen in subsection \ref{Choice-gauge}, we can follow word by word the
proof given in Lemma 5 of the paper \cite{GN10}.

Let us assume by absurd that \rf{Co-F string} is not satisfied, then 
\rf{Inter-step} implies that $\text{\large s}_{p,\la_0}\subset\text{\large{Z}}_{\text{\textsc{C}}_{1,2l+1}(\la)}$. Now this last condition implies that \rf{Bax-eq} can be satisfied only if $\text{\large s}_{p,\la_0}\subset\text{\large{Z}}_{{\tt d}(\la)}$, which does not hold by the definition of ${\tt d}(\la)$.
\end{proof}

\section{Asymptotics of Yang-Baxter generators}\label{Asymp-A-D}

\setcounter{equation}{0}

\renewcommand{\su}{{\mathsf u}}
\renewcommand{\sv}{{\mathsf v}}

From the known form of the Lax operator we derive the following asymptotics
for $\lambda \rightarrow +\infty $ and $0$ of the generators of the
Yang-Baxter algebras.
\begin{itemize}
\item[] \textbf{\SRN \ odd:} The leading operators are $\SB_{\SRN}(\lambda )$
and $\SC_{\SRN}(\lambda )$\ with asymptotics:
\begin{align}
\SB_{\SRN}(\lambda )& =\left( \prod_{a=1}^{\SRN}\frac{\kappa _{a}}{i}\right) \left(
\lambda ^{\SRN}\prod_{a=1}^{\SRN}\frac{\sv_{a}^{(-1)^{1+a}}}{\xi _{a}}-\lambda
^{-\SRN}\prod_{a=1}^{\SRN}\xi _{a}\sv_{a}^{(-1)^{a}}\right) +\text{sub-leading}, \\
\SC_{\SRN}(\lambda )& =\left( \prod_{a=1}^{\SRN}\frac{\kappa _{a}}{i}\right) \left(
\lambda ^{\SRN}\prod_{a=1}^{\SRN}\frac{\sv_{a}^{(-1)^{a}}}{\xi _{a}}-\lambda
^{-\SRN}\prod_{a=1}^{\SRN}\xi _{a}\sv_{a}^{(-1)^{1+a}}\right) +\text{sub-leading}.
\end{align}
\item[] \textbf{\SRN \ even:} The leading operators are $\SA_{\SRN}(\lambda )$
and $\SD_{\SRN}(\lambda )$\ with asymptotics:
\begin{align}
\SA_{\SRN}(\lambda )& =\left( \prod_{a=1}^{\SRN}\frac{\kappa _{a}}{i}\right) \left(
\lambda ^{\SRN}\prod_{a=1}^{\SRN}\frac{\sv_{a}^{(-1)^{a}}}{\xi _{a}}+\lambda
^{-\SRN}\prod_{a=1}^{\SRN}\xi _{a}\sv_{a}^{(-1)^{1+a}}\right) +\text{sub-leading},
\label{asymp-A} \\
\SD_{\SRN}(\lambda )& =\left( \prod_{a=1}^{\SRN}\frac{\kappa _{a}}{i}\right) \left(
\lambda ^{\SRN}\prod_{a=1}^{\SRN}\frac{\sv_{a}^{(-1)^{1+a}}}{\xi _{a}}+\lambda
^{-\SRN}\prod_{a=1}^{\SRN}\xi _{a}\sv_{a}^{(-1)^{a}}\right) +\text{sub-leading}.
\end{align}
\end{itemize}
Note that these asymptotics imply for the SOV representation of the
Yang-Baxter generators the following formulae\footnote{%
Here, we denote with $\mathsf{w}^{\text{SOV}}$ the invertible matrix which defines the change of basis in the representation from the basis obtained by the tensor product of the $\sv_n$-eigenbasis to the $\SB$-eigenbasis.}:
\begin{itemize}
\item[] \textbf{\SRN \ odd:}
\begin{equation}
\left( \mathsf{w}^{\text{SOV}}\right) ^{-1}\left(
\prod_{a=1}^{\SRN}\sv_{a}^{(-1)^{1+a}}\right) \mathsf{w}^{\text{SOV}%
}=\prod_{a=1}^{\SRN}\frac{\xi _{a}}{\eta _{a}}.
\end{equation}
\item[] \textbf{\SRN \ even:}
\begin{eqnarray}
\prod_{a=1}^{\SRN}\xi _{a}\left( \mathsf{w}^{\text{SOV}}\right) ^{-1}\Theta \mathsf{w}^{\text{SOV}} &=&\left( \eta _{A}\prod_{a=1}^{\SRN-1}\eta
_{a}\right) \text{\textsf{T}}_{\SRN}^{-}, \\
\prod_{a=1}^{\SRN}\xi _{a}\left( \mathsf{w}^{\text{SOV}}\right) ^{-1}\Theta
^{-1}\mathsf{w}^{\text{SOV}} &=&\left( \eta _{D}\prod_{a=1}^{\SRN-1}\eta
_{a}\right) \text{\textsf{T}}_{\SRN}^{+},
\end{eqnarray}
\end{itemize}
Note that taking the average value of the last two formulae we get for $\SRN$ odd:
\begin{equation}
\prod_{a=1}^{\SRN}\frac{X_{a}}{Z_{a}}=\prod_{a=1}^{\SRN}V_{a}^{(-1)^{1+a}},
\end{equation}%
while for $\SRN$ even:
\begin{equation}\label{ZAD}
Z_{\SA}=\langle \Theta \rangle \prod_{a=1}^{\SRN-1}Z_{a}^{-1}\prod_{a=1}^{\SRN}X_{a},\text{ \ \ \ \ }%
Z_{\SD}=Z_{\SA}\langle \Theta \rangle ^{-2},
\end{equation}
where
\begin{equation}
\langle \Theta \rangle=\prod_{a=1}^{\SRN}V_{a}^{(-1)^{1+a}},
\end{equation}
is the average value of the charge $\Theta$.

\section{Comparison between the $\tau_2$-model and the SG model}

\setcounter{equation}{0}

In this appendix we present the comparison between the SG model and the
so-called $\tau_2$-model introduced and analyzed in the series of papers \cite{BS, BBP90}-\cite{GIPST08}.

\subsection{Lax operators}

The Lax operator which describes the $\tau_2$-model reads:%
\begin{equation}
L_{n}^{\tau_2}\equiv \text{\textsc{A}}_{n}\left( 
\begin{array}{cc}
D_{n}\equiv (d_{n,+}\text{\textsc{X}}_{n}+d_{n,-}\text{\textsc{X}}_{n}^{-1})
& H_{n}\equiv \left( h_{n,-}\text{\textsc{X}}_{n}+h_{n,+}\text{\textsc{X}}%
_{n}^{-1}\right) \text{\textsc{Y}}_{n} \\ 
G_{n}\equiv \left( g_{n,-}\text{\textsc{X}}_{n}+g_{n,+}\text{\textsc{X}}%
_{n}^{-1}\right) \text{\textsc{Y}}_{n}^{-1} & F_{n}\equiv (f_{n,-}\text{%
\textsc{X}}_{n}+f_{n,+}\text{\textsc{X}}_{n}^{-1})%
\end{array}%
\right) \text{\textsc{A}}_{n}^{-1},  \label{L-t2-LSG}
\end{equation}%
where:%
\begin{equation}
\text{\textsc{A}}_{n}\equiv \left( 
\begin{array}{cc}
u_{n} & 0 \\ 
0 & 1%
\end{array}%
\right) 
\end{equation}%
and: 
\begin{equation}
g_{n,+}=f_{n,+}d_{n,-}/h_{n,+},\hspace{2cm}g_{n,-}=f_{n,-}d_{n,+}/h_{n,-}.
\end{equation}%
Here, we have denoted with \textsc{Y}$_{n}\equiv \su_{n}/u_{n}$ and \textsc{X%
}$_{n}\equiv \sv_{n}/v_{n}$ the generators of the local Weyl algebra \textsc{%
Y}$_{n}$\textsc{X}$_{n}=q$\textsc{X}$_{n}$\textsc{Y}$_{n}$ with the
following central values \textsc{X}$_{n}^{p}=$\textsc{Y}$_{n}^{p}=1.$

Then, the Lax operator of the $\tau_2$-model can be related to one of
the SG model by: 
\begin{equation}
L_{n}^{\tau_2}=L_{n}^{SG}\sigma _{1},  \label{Lax-t2-SG}
\end{equation}%
when the parameters in (\ref{L-t2-LSG}) are identified by:%
\begin{eqnarray}
d_{n,-} &=&d_{n,+}^{-1},\,\quad \,\quad \,\quad \,\,\,\,f_{n,-}=f_{n,+}^{-1},
\label{U_q} \\
f_{n,+} &=&-i\lambda _{n}v_{n}^{-1},\,\quad \quad \,d_{n,+}=-i\lambda
_{n}v_{n}, \\
h_{n,-} &=&q^{1/2}\kappa _{n}v_{n},\,\quad h_{n,+}=q^{-1/2}\kappa
_{n}^{-1}v_{n}^{-1},  \label{par-h} \\
g_{n,+} &=&q^{1/2}\kappa _{n}v_{n}^{-1},\,\quad g_{n,-}=q^{-1/2}\kappa
_{n}^{-1}v_{n}.  \label{par-g}
\end{eqnarray}%
It is worth pointing out that in our SG model, we are asking that the
elements of the Lax operator generate the algebra $U_{q}(sl_{2})$ and that
such a requirement implies the relations \rf{U_q}.

It is worth pointing out the different role played by the parameters $u_{n}$
in these two models. From (\ref{L-t2-LSG}), we see that the parameter $%
u_{n}$ enters in the $\tau_2$-Lax operator as a simple gauge while this
is not the case for the SG Lax operator being\footnote{%
The presence of the $\sigma _{1}$ tells us that what is a simple gauge for
the $\tau_2$-model is not a gauge for the SG model and vice versa.} $%
[\sigma _{1},$\textsc{A}$_{n}]\neq 0$. In particular, this means that while
we are free to change uniformly along the chain the average values of the
local generators $\su_{n}$ without any modification in the $\tau_2$%
-transfer matrix this is not the case for the SG transfer matrix.

\subsection{Map from the SG model to the $\protect\tau_2$-model for
even chain}

The relation (\ref{Lax-t2-SG}) directly implies that for an odd chains the
transfer matrices in the two models have different spectrum. In the case of
even chain a more precise analysis is need and it is here developed.

Let us fix $\SRN=2\SRM$ with $\SRM$ positive integer; then, it is simple to
define a map from the generators of the Yang-Baxter algebras in the SG model
to those in the $\tau_2$-model. We have just to point out that the
relation \rf{L-t2-LSG} implies: 
\begin{equation}
\SM^{\tau_2}(\lambda )=\sigma _{1}\Pi _{\tau_2}\left( \SM%
^{SG}(\lambda )\right) \sigma _{1}\,\,\,\longrightarrow \,\,\,\ST^{\tau_2}(\lambda )=\Pi _{\tau_2}\left( \ST^{SG}(\lambda )\right) .
\end{equation}%
where we have defined: 
\begin{equation}
\Pi _{\tau_2}\left( L_{2m-a}^{SG}(\lambda _{2m-a})\right) \equiv \left(
\sigma _{1}\right) ^{1-a}L_{2m-a}^{SG}((-1)^{(1-a)}\lambda _{2m-a})\left(
\sigma _{1}\right) ^{1-a}=\left\{ 
\begin{array}{l}
L_{2m-1}^{SG}(\lambda _{2m-1}),\vspace{0.3cm} \\ 
\left( L_{2m}^{SG}(-\varepsilon \lambda _{2m}^{\ast })\right) ^{\dagger },%
\end{array}%
\right.
\end{equation}%
for $m\in \{1,...,\SRM\}\,\,\text{and }\,a\in \{0,1\}$; note that in the
last equality we have used \rf{Hermit-L}. Then, we can characterize the map $%
\Pi _{\tau_2}$\ by its action on the generators of the local Weyl
algebras:%
\begin{equation}
\Pi _{\tau_2}\equiv \prod\limits_{n=1}^{\SRM}\digamma _{2n},\,\,\,\,%
\text{ where:}\,\,\,\,\,\,\digamma _{n}(\su_{n})=\su_{n}^{-1},\text{ \ \ }%
\digamma _{n}(\sv_{n})=\sv_{n}^{-1}.
\end{equation}%
Until now, the discussion presented is common for both the untwisted and
twisted representations of the SG model. However, the distinction between
these two cases appears evident when we observe that:

\begin{description}
\item[Untwisted representations:] The map $\Pi _{\tau_2}$ \textbf{is}
realized by the unitary transformation $\pi _{\tau_2}$:%
\begin{equation}
\Pi _{\tau_2}(\SM^{SG}(\la))=\pi _{\tau ^{2}}\SM^{SG}(\la)\pi _{\tau ^{2}},\,\,\,\,\text{ }\,\,\pi _{\tau ^{2}}\equiv
\prod\limits_{n=1}^{\SRM}\Omega_{2n},
\end{equation}
after the flipping $\xi_{2n-a}\,\rightarrow\,(-\varepsilon\, )^{(1-a)}\xi_{2n-a}$ of the inhomogeneities. Here, $\varepsilon$ is the sign defined in \rf{cond-T-Normality} and $\Omega_{n}$ are the local unitary transformations defined by the
following action:%
\begin{equation*}
\text{ \ }\Omega_{n}|z_{1},...,z_{N}\rangle \equiv
|z_{1},...,z_{n}^{-1},...,z_{N}\rangle
\end{equation*}%
on the elements of the basis \rf{u-basis} of the space of the representation.

\item[Twisted representations:] The map $\Pi _{\tau_2}$ \textbf{cannot}
be realized by any similarity transformation on the space of the
representation.
\end{description}

Indeed, the distinction between these two cases can be simply understood by
the identities:%
\begin{equation}
\digamma_{n}(\su_{n})=\left(\Omega_{n}\su_{n}\Omega_{n}\right) /u_{n}^{2},%
\text{ \ \ }\digamma_{n}(\sv_{n})=\left(\Omega_{n}\sv_{n}\Omega_{n}\right)
/v_{n}^{2}.
\end{equation}

which tell us that for untwisted representations the local maps $\digamma
_{n}$ leave unchanged the spectrum of the generators of the Weyl algebras
while this is not anymore true for the twisted ones.

In particular, these statements imply that for even chain and
untwisted representations the SG transfer matrix is similar (up to the inhomogeneity flipping) to the $\tau_2$-transfer matrix while this is not anymore the case for twisted
representations.

\textbf{Remark 6.} It is worth remarking that while the spectral problem of
the SG model and the $\tau_2$-model are in general different the method
here introduced to analyzed the SG spectrum can be extended to $\tau_2$%
-model. Such an extension is of sure interest as it should allow to get also for the $\tau_2$-model the strong statements of the simplicity and completeness of the transfer matrix spectrum. 


\begin{thebibliography}{99}
\begin{small}
\newcommand{\CMP}[3]{{Comm. Math. Phys. }{\bf #1} (#2) #3}
\newcommand{\LMP}[3]{{Lett. Math. Phys. }{\bf #1} (#2) #3}
\newcommand{\IMP}[3]{{Int. J. Mod. Phys. }{\bf A#1} (#2) #3}
\newcommand{\NP}[3]{{Nucl. Phys. }{\bf B#1} (#2) #3}
\newcommand{\PL}[3]{{Phys. Lett. }{\bf B#1} (#2) #3}
\newcommand{\MPL}[3]{{Mod. Phys. Lett. }{\bf A#1} (#2) #3}
\newcommand{\PRD}[3]{{Phys. Rev.}{\bf D#1} (#2) #3}
\newcommand{\PRL}[3]{{Phys. Rev. Lett. }{\bf #1} (#2) #3}
\newcommand{\AP}[3]{{Ann. Phys. (N.Y.) }{\bf #1} (#2) #3}
\newcommand{\LMJ}[3]{{Leningrad Math. J. }{\bf #1} (#2) #3}
\newcommand{\FAA}[3]{{Funct. Anal. Appl. }{\bf #1} (#2) #3}
\newcommand{\PTPS}[3]{{Progr. Theor. Phys. Suppl. }{\bf #1} (#2) #3}
\newcommand{\LMN}[3]{{Lecture Notes in Mathematics }{\bf #1} (#2) #2}

\bibitem{NT} G. Niccoli and J. Teschner, {\it The Sine-Gordon model revisited I}, J. Stat. Mech. (2010) P09014.
\bibitem{GN10} G. Niccoli, \textit{Reconstruction of Baxter
Q-operator from Sklyanin SOV for cyclic representations of integrable
quantum models} , Nucl. Phys. \textbf{B} 835
[PM] (2010) 263-283.

\bibitem{FST} L.D. Faddeev, E.K. Sklyanin, L.A. Takhtajan,
{\it Quantum inverse problem method: I},
Theor. Math. Phys. {\bf 57} (1980) 688-706

\bibitem{TTF83} V.O. Tarasov, I,. A. Takhtadzhyan and L.D. Faddeev, \textit{Local Hamiltonian for Integrable Quantum Models on a Lattice}, Theo. Math. Phys. {\bf 57}, 2 (1983) 1059-1073

\bibitem{IK} A.G. Izergin, V.E. Korepin,
{\it Lattice versions of quantum field theory models in two
dimensions}, Nucl. Phys. {\bf B205} (1982) 401-413

\bibitem{FV94} L.D.  Faddeev, A. Yu. Volkov,
{\it Hirota Equation as an Example of an Integrable Symplectic Map},
Letters in Mathematical Physics {\bf 32} (1994) 125-135

\bibitem{BBR}
V. Bazhanov, A. Bobenko, N. Reshetikhin,
{\it Quantum discrete sine-Gordon model at roots of $1$:
integrable quantum system on the integrable classical background.}
Comm. Math. Phys.  {\bf 175}  (1996),  no. 2, 377-400

\bibitem{Ba08} V.V. Bazhanov,
{\it  Chiral Potts model and the discrete Sine-Gordon model at
roots of unity}, Preprint arXiv:hep-th/0809.2351

\bibitem{FT79} L.D. Faddeev and L.A. Takhtajan, {\it The quantum method of the inverse problem and the
Heisenberg XYZ-model}, Russ. Math. Surveys, 34:5 (1979) 11–68

\bibitem{KS79} P.P. Kulish and E.K. Sklyanin, {\it Quantum inverse scattering method and the Heisenberg ferromagnet}, Phys. Lett. {\bf A70} (1979) 461-463

\bibitem{ADMcCoy92}  G. Albertini, S. Dasmahapatra and B. McCoy, {\it Spectrum and Completeness of the Integrable 3-State Potts Model: A Finite Size Study}, Int. J. Mod. Phys. {\bf A7}, Suppl. 1A (1992) 1-53

\bibitem{FMcCoy01} K. Fabricius and B. McCoy, {\it Bethe's Equation Is Incomplete for the XXZ Model at Roots of Unity}, J. Stat. Phys. {\bf 103} (2001) 647-678

\bibitem{NR03} R. I. Nepomechie and F. Ravanini, {\it Completeness of the Bethe Ansatz solution of the open XXZ chain with nondiagonal boundary terms}, J. Phys. {\bf A36} (2003) 11391-11402

\bibitem{MTV} E. Mukhin, V. Tarasov, A. Varchenko,
{\it Bethe Algebra of Homogeneous XXX  Heisenberg Model has Simple Spectrum.}
Comm. Math. Phys. {\bf 288} (2009) 1-42
 
\bibitem{Sk1}
E.K.~Sklyanin, {\em The quantum Toda chain},
 Lect.\ Notes Phys.\ {\bf 226} (1985) 196--233

\bibitem{Sk2}
E.K.~Sklyanin, {\em Quantum inverse scattering method.
 Selected topics}. In: {\em Quantum groups and quantum
integrable systems} (World Scientific, 1992) 63--97

\bibitem{Sk3}
E.K.~Sklyanin,
 {\em Separation of variables -- new trends},
Prog.\ Theor.\ Phys.\ Suppl.\  {\bf 118} (1995) 35--60

\bibitem{BLZ-I} V.V. Bazhanov, S.L. Lukyanov and A.B. Zamolodchikov,
{\it Integrable Structure of Conformal Field Theory, Quantum KdV Theory and
Thermodynamic Bethe Ansatz}, Commun. Math. Phys. 177 (1996) 381-398.

\bibitem{BLZ-II} V.V. Bazhanov, S.L. Lukyanov and A.B. Zamolodchikov,
{\it Integrable Structure of Conformal Field Theory II. Q-operator and DDV
equation}, Commun. Math. Phys. 190 (1997) 247-278.

\bibitem{BLZ-III} V.V. Bazhanov, S.L. Lukyanov and A.B. Zamolodchikov,
{\it Integrable Structure of Conformal Field Theory III. The Yang-Baxter
Relation}, Commun. Math. Phys. 200 (1999) 297-324.

\bibitem{BLZ-IV} V.V. Bazhanov, S.L. Lukyanov and A.B. Zamolodchikov,
{\it Quantum field theories in finite volume: excited state energies}, Nucl.
Phys. B 489 (1997) 487-531.

\bibitem{AF96} A. Antonov and B. Feigin, {\it Quantum Group Representations
and Baxter Equation}, Phys. Lett. B 392 (1997), 115-122.

\bibitem{RW02} M. Rossi and R. Weston, {\it A generalized Q-operator for
$U_{q}(sl_2)$ vertex models}, J. Phys. A 35 (2002) 10015-10032.

\bibitem{Ba73} R. Baxter,
{\it  Eight-vertex model in lattice statistics and one-dimensional
anisotropic Heisenberg chain. I. Some fundamental eigenvectors}
Annals of Physics {\bf 76} (1973) 1-24

\bibitem{BS} V.V. Bazhanov, Yu. G. Stroganov,
{\it Chiral Potts model as a descendant of the six-vertex model},
Journal of Statistical Physics {\bf 59} (1990) 799-817

\bibitem{GP}
V. Pasquier and M. Gaudin, {\it The periodic Toda chain and a matrix generalization of the Bessel function},
J. Phys. {\bf A25} (1992) 5243-5252

\bibitem{De} S.E. Derkachov, {\it Baxter’s Q-operator for the homogeneous XXX spin chain}, J. Phys. {\bf A32} (1999) 5299-5316

\bibitem{DKM} S.E. Derkachov, G.P. Korchemsky, and A. N. Manashov, {\it Noncompact Heisenberg spin
magnets from high–energy QCD. I. Baxter Q-operator and separation of variables}, Nucl. Phys. {\bf B617} (2001) 375-440; {\it Separation of variables for the quantum $SL(2,{\Bbb{R}})$ spin chain}, JHEP 0307 (2003) 047

\bibitem{Be31} H. Bethe, {\it  Zur Theorie der Metalle I. Eigenwerte und Eigenfunktionen der linearen Atomkette}, Z. Phys. {\bf 71} (1931) 205

\bibitem{Bax82} R.J. Baxter, {\it Exactly Solved Models in Statistical Mechanics}, Academic Press, 1982

\bibitem{ABBQ87} F.C. Alcaraz, M.N. Barber, M.T. Batchelor, R.J. Baxter and G.R.W. Quispel, {\it Surface exponents of the quantum XXZ, Ashkin-Teller and Potts models}, J. Phys. {\bf A20} (1987) 6397-6409

\bibitem{R83} N.Yu. Reshetikhin, {\it A Method Of Functional Equations In The Theory Of Exactly Solvable Quantum Systems}, Lett. Math. Phys. {\bf 7} (1983) 205-213; Sov. Phys. JETP {\bf 57} (1983) 691

\bibitem{FK10} G. Filali and N. Kitanine,
{\em Spin Chains with Non-Diagonal Boundaries and Trigonometric SOS Model with Reflecting End},
SIGMA 7 (2011), 012

\bibitem{Ta} V. Tarasov, {\it
Cyclic monodromy matrices for the $R$-matrix of the six-vertex model
and the chiral Potts model with fixed spin boundary conditions.}
Infinite analysis, Parts A, B (Kyoto, 1991),
963--975, Adv. Ser. Math. Phys., 16, World Sci. Publ., River Edge, NJ, 1992.

\bibitem{FK2} L.D. Faddeev and R.M. Kashaev, {\it Quantum dilogarithm}, Mod. Phys. Lett. {\bf A9} (1994) 427-434

\bibitem{F2} L.D. Faddeev, {\it Discrete Heisenberg-Weyl group and modular group}, Lett. Math. Phys. {\bf 34} (1995) 249-254

\bibitem{Ru} S.N.M. Ruijsenaars, {\it  First order analytic difference equations and integrable quantum systems}, J. Math.
Phys. {\bf 38} (1997) 1069-1146

\bibitem{Wo} S.L.Woronowicz, {\it Quantum exponential function}, Rev. Math. Phys. {\bf 12} (2000) 873-920

\bibitem{PT2} B. Ponsot and J. Teschner, {\it Clebsch–Gordan and Racah–Wigner coefficients for a continuous series of
representations of $\ {\cal U}_q (sl(2,{\Bbb R}))$}, Commun. Math. Phys. {\bf 224} (2001) 613-655

\bibitem{K1} R.M. Kashaev, {\it The non–compact quantum dilogarithm and the Baxter equations}, J. Stat. Phys. {\bf 102}
(2001) 923-936

\bibitem{K2} R.M. Kashaev, {\it The quantum dilogarithm and Dehn twists in quantum Teichm\"uller theory}, In: Integrable
structures of exactly solvable two–dimensional models of quantum field theory (Kiev, 2000), 211-221
(NATO Sci.Ser.II Math.Phys.Chem., 35, Kluwer Acad. Publ., Dordrecht, 2001)

\bibitem{BT03} A. Bytsko and J. Teschner, {\it R-operator, co–product and Haar-measure for the modular double of
${\cal U}_q(sl(2,{\Bbb R}))$}, Commun. Math. Phys. {\bf 240} (2003) 171–196

\bibitem{T2} J. Teschner, {\it Liouville theory revisited}, Class. Quant. Grav. {\bf 18} (2001) R153-R222; {\it A lecture on the Liouville vertex operators}, Int. J. Mod. Phys. {\bf A19S2} (2004) 436-458

\bibitem{V2} A.Yu. Volkov, {\it Noncommutative hypergeometry}, Commun. Math. Phys. {\bf 258} (2005) 257-273

\bibitem{BT06} A. Bytsko, J. Teschner, {\it Quantization of models with non-compact quantum group symmetry.
Modular XXZ magnet and lattice sinh-Gordon model}, J. Phys. {\bf A39} (2006) 12927-12981

\bibitem{IK81} A. G. Izergin and V. E. Korepin, {\it A lattice model connected with the nonlinear Schr¨odinger
equation} (in Russian), Doklady Akademii Nauk {\bf 259} (1981) 76

\bibitem{IK09} A. G. Izergin and V. E. Korepin, {\it A lattice model related to the nonlinear Schroedinger equation}, arXiv:0910.0295 

\bibitem{BR89} V.V. Bazhanov and N.Yu. Reshetikhin, {\it Critical RSOS model and conformal field theory} Int. J. Mod. Phys. {\bf A4}, (1989) 115-142 

\bibitem{Ne02} R.I. Nepomechie, {\it Functional relations and Bethe Ansatz for the XXZ chain}, J. Stat. Phys. {\bf 111},  (2003) 1363-1376

\bibitem{Ne03} R.I. Nepomechie,  {\it Bethe ansatz solution of the open XXZ chain with nondiagonal boundary terms}, J. Phys. {\bf A37}, (2004) 433-440

\bibitem{BBP90} R.J. Baxter, V.V. Bazhanov and J.H.H. Perk, {\it Functional relations for the transfer matrices of the Chiral Potts model}, Int. J. Mod. Phys. {\bf B4} (1990) 803-869

\bibitem{Ba89} R.J. Baxter, {\it Superintegrable chiral Potts model: thermodynamic properties, an ‘Inverse’ model, and a simple associated Hamiltonian}, J. Stat. Phys. {\bf 57} (1989) 1–39

\bibitem{Ba04} R.J. Baxter {\it Transfer Matrix Functional Relations for the Generalized $\tau2(tq)$ Model}, J. Stat. Phys. {\bf 117} (2004) 1-25

\bibitem{GIPS06} G. von Gehlen, N. Iorgov, S. Pakuliak and V. Shadura, {\it The Baxter-Bazhanov-Stroganov model: separation of variables and the Baxter equation}, J. Phys. A: Math. Gen. {\bf 39} (2006) 7257-7282

\bibitem{GIPST07} G. von Gehlen, N. Iorgov, S. Pakuliak, V. Shadura and Yu Tykhyy, {\it Form-factors in the Baxter-Bazhanov-Stroganov model I: norms and matrix elements}, J. Phys. A: Math. Theor. {\bf 40} (2007) 14117-14138

\bibitem{GIPST08} G. von Gehlen, N. Iorgov, S. Pakuliak, V. Shadura and Yu Tykhyy, {\it Form-factors in the Baxter-Bazhanov-Stroganov model II: Ising model on the finite lattice}, J. Phys. A: Math. Theor. {\bf 41} (2008) 095003 (24pp)

\bibitem{GM} M. Gutzwiller, {\it The quantum mechanical Toda lattice II}, Ann. of Phys. {\bf 133} (1981), 304-331.

\bibitem{KL99} S. Kharchev, D. Lebedev, {\it Integral representation for the eigenfunctions of quantum periodic Toda chain}, Lett.Math.Phys. {\bf 50} (1999) 53-77


\bibitem{KP91} A. Kl\"umper, P.A. Pearce, {\it Analytic calculation of scaling dimensions: Tricritical hard squares and critical hard hexagons}, J. Stat. Phys. {\bf 64} (1991) 13-76

\bibitem{KBP91} A. Kl\"umper, M. Batchelor, P.A. Pearce, {\it Central charges of the 6- and 19-vertex models with twisted boundary conditions}, J. Phys. {\bf A23} (1991) 3111-3133

\bibitem{KT10} K. K. Kozlowski, J. Teschner, {\it TBA for the Toda chain}, arXiv:1006.2906v1

\bibitem{Za} Al.B. Zamolodchikov, {\it On the Thermodynamic Bethe Ansatz Equation in Sinh-Gordon
Model} J. Phys. {\bf A39}, (2006), 12863-12887.

\bibitem{T} J. Teschner,
{\it On the spectrum of the Sinh-Gordon model in finite volume},
Nucl.Phys.{\bf B799} (2008) 403-429

\bibitem{DDV92} C. Destri, H.J. De Vega, {\it New thermodynamic Bethe ansatz equations without strings}, Phys. Rev. Lett. {\bf 69} (1992) 2313-2317

\bibitem{DDV94} C. Destri, H.J. de Vega,
{\it  Unified Approach to Thermodynamic Bethe Ansatz and Finite Size Corrections for Lattice Models and Field Theories}, Nucl. Phys. {\bf B438} (1995) 413-454

\bibitem{DDV97} C. Destri, H.J. de Vega,
{\it  Non linear integral equation and excited--states scaling functions in the sine-Gordon model},
Nucl. Phys. {\bf B504} (1997) 621-664

\bibitem{FMQR96} D. Fioravanti, A. Mariottini, E.Quattrini, F. Ravanini,
{\it  Excited state Destri-De Vega equation for sine-Gordon and restricted sine-Gordon},
Phys. Lett. {\bf B390} (1997) 243-251

\bibitem{FRT98} G. Feverati, F. Ravanini, G. Takacs,
{\it  Truncated conformal space at c=1, nonlinear integral equation and quantization rules for multi-soliton states},
Phys. Lett. {\bf B430} (1998) 264-273

\bibitem{FRT99} G. Feverati, F. Ravanini, G. Takacs
{\it  Nonlinear Integral Equation and Finite Volume Spectrum of Sine-Gordon Theory},
Nucl. Phys. {\bf B540} (1999) 543-586

\bibitem{F00} G. Feverati, {\it  Finite Volume Spectrum of Sine-Gordon Model and its Restrictions}, Ph.D. Thesis, Bologna University (2000), hep-th/0001172

\bibitem{R01} F. Ravanini, {\it  Finite Size Effects in Integrable Quantum Field Theories}, hep-th/0102148

\bibitem{FR02I} D. Fioravanti, M. Rossi, {\it From the braided to the usual Yang-Baxter relation}, J. Phys. {\bf A34} (2001) L567-L576 

\bibitem{FR02II} D. Fioravanti, M. Rossi, {\it A Braided Yang-Baxter algebra in a theory of two coupled
 lattice quantum KdV: Algebraic properties and ABA representations}, J. Phys. {\bf A35} (2002) 3647-3682

\bibitem{FR03I} D. Fioravanti, M. Rossi, {\it Exact conserved quantities on the cylinder 1: Conformal case}, JHEP 0307:031, (2003)
 
\bibitem{FR03II} D. Fioravanti, M. Rossi, {\it Exact conserved quantities on the cylinder. 2. Off critical case}, JHEP 0308:042, (2003)

\bibitem{FV} L.D.  Faddeev, A. Yu. Volkov,
{\it Quantum inverse scattering method on a space-time lattice},
 Theor. Math. Phys. {\bf  92} (1992) 837-842

\bibitem{BKP93} A. Bobenko, N. Kutz and U. Pinkall, 
{\it The discrete quantum pendulum}, Phys. Lett. {\bf A177} (1993) 399-404
\end{small}
\end{thebibliography}
\end{document}